\documentclass[journal,10pt]{IEEEtran}
\usepackage[utf8]{inputenc}
\usepackage{amsthm}

\usepackage{amsmath}
\usepackage{amssymb}
\usepackage{amsfonts}
\usepackage{cite}

\usepackage{mathrsfs}
\usepackage{booktabs}
\usepackage{bbm}

\usepackage{array}
\usepackage{makecell}
\usepackage{multirow}
\usepackage{diagbox}
\usepackage{pifont}

\DeclareMathOperator*{\argmin}{arg\,min}

\newcommand{\Rmnum}[1]{\expandafter\@slowromancap\romannumeral #1@}
\usepackage{graphicx}
\usepackage{subfigure}
\usepackage{epsfig}
\usepackage[numbers,sort&compress]{natbib}
\usepackage{xcolor}
\usepackage{color}
\usepackage{enumerate}
\usepackage{extarrows}

\usepackage[ruled,linesnumbered,lined]{algorithm2e} 
\usepackage{algorithmic}

\usepackage{tensor}

\usepackage[colorlinks,linkcolor=blue,anchorcolor=blue,citecolor=blue,urlcolor=black]{hyperref}
\usepackage{url}

\begin{document}
\title{Neuromorphic Split Computing with Wake-Up Radios: Architecture and Design  via Digital Twinning}
\author{Jiechen Chen, \IEEEmembership{Member,~IEEE}, Sangwoo Park,  \IEEEmembership{Member,~IEEE}, Petar Popovski,~\IEEEmembership{Fellow,~IEEE}, H. Vincent Poor,  \IEEEmembership{Life Fellow,~IEEE}, Osvaldo Simeone,~\IEEEmembership{Fellow,~IEEE}
\thanks{J. Chen, S. Park, and O. Simeone are with the King’s Communications, Learning and Information Processing (KCLIP) lab within the Centre for Intelligent Information Processing Systems (CIIPS) at the Department of Engineering, King’s College London, London, WC2R 2LS, UK (email:\{jiechen.chen, sangwoo.park, osvaldo.simeone\}@kcl.ac.uk). O. Simeone is also with the Department of Electronic Systems, Aalborg University, 9100 Aalborg, Denmark. P. Popovski is with the Department of Electronic Systems, Aalborg University, 9100 Aalborg, Denmark (email: petarp@es.aau.dk). H. Vincent Poor is with the Department of Electrical and Computer Engineering, Princeton University, Princeton, NJ 08544 USA (e-mail:poor@princeton.edu).\\
This work was supported by the European Union’s Horizon Europe project CENTRIC (101096379),  by~an Open Fellowship of the EPSRC (EP/W024101/1), by the EPSRC project (EP/X011852/1), by the Villum Investigator Grant ``WATER" from the Velux Foundations, Denmark, and by the U.S. National Science Foundation under Grant ECCS-2335876. 
 }
\vspace*{-0.9cm}
}

\maketitle

\vspace{-1.7cm}
\begin{abstract}
    Neuromorphic computing leverages the sparsity of temporal data to reduce processing energy by activating a small subset of neurons and synapses at each time step. When deployed for split computing in edge-based systems, remote neuromorphic processing units (NPUs) can reduce the communication power budget by communicating asynchronously using sparse impulse radio (IR) waveforms. This way, the input signal sparsity translates directly into energy savings both in terms of computation and communication. However, with IR transmission, the main contributor to the overall energy consumption remains the power required to maintain the main radio on. This work proposes a novel architecture that integrates a wake-up radio mechanism within a split computing system consisting of remote, wirelessly connected, NPUs. A key challenge in the design of a wake-up radio-based neuromorphic split computing system is the selection of thresholds for sensing, wake-up signal detection, and decision making. To address this problem, as a second contribution, this work proposes a novel methodology that leverages the use of a digital twin (DT), i.e., a simulator, of the physical system, coupled with a sequential statistical testing approach known as Learn Then Test (LTT) to provide theoretical reliability guarantees. The proposed DT-LTT methodology is broadly applicable to other design problems, and is showcased here for neuromorphic communications.  Experimental results validate the design and the analysis, confirming the theoretical reliability guarantees and illustrating trade-offs among reliability, energy consumption, and informativeness of the decisions.
\end{abstract}


\IEEEpeerreviewmaketitle
\newtheorem{definition}{\underline{Definition}}[section]
\newtheorem{fact}{Fact}
\newtheorem{assumption}{Assumption}
\newtheorem{theorem}{Theorem}
\newtheorem{lemma}{\underline{Lemma}}[section]
\newtheorem{proposition}{\underline{Proposition}}[section]
\newtheorem{corollary}[proposition]{\underline{Corollary}}
\newtheorem{example}{\underline{Example}}[section]
\newtheorem{remark}{\underline{Remark}}[section]
\newcommand{\mv}[1]{\mbox{\boldmath{$ #1 $}}}
\newcommand{\mb}[1]{\mathbb{#1}}
\newcommand{\Myfrac}[2]{\ensuremath{#1\mathord{\left/\right.\kern-\nulldelimiterspace}#2}}
\newcommand\Perms[2]{\tensor[^{#2}]P{_{#1}}}
\newcommand{\note}[1]{[\textcolor{red}{\textit{#1}}]}

\section{Introduction}
\subsection{Context and Motivation}

 Neuromorphic processing units (NPUs), such as Intel's Loihi or BrainChip's Akida,  leverage the sparsity of temporal data to reduce processing energy by activating a small subset of neurons and synapses at each time step \cite{davies2021advancing, jang2019introduction}. This mechanism implements \emph{spike}-based signaling, whereby information is exchanged in the timing of the synaptic activation.  The opportunistic activation of neurons and synapses distinguishes NPUs from conventional deep learning accelerators such as graphical processing units (GPUs) or tensor processing units (TPUs), making NPUs particularly attractive for time-series data.

 \begin{figure}[t!]
	\centering
	\includegraphics[width=3.5in]{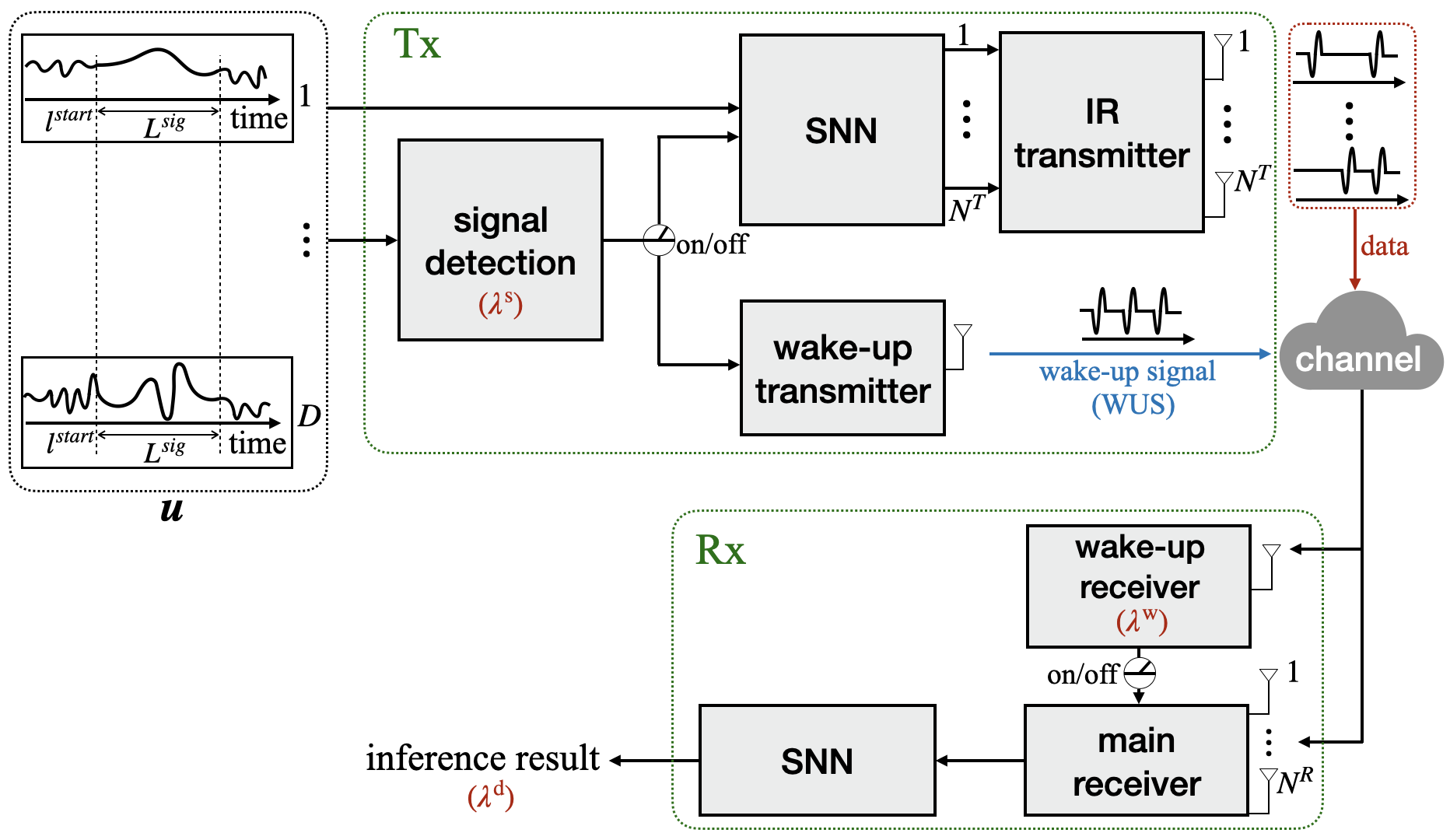}
	\caption{In this work, we propose a low-power wake-up radio aided wireless split computing system, which operates through the following steps. (\emph{i}) \emph{Signal detection at the Tx}: The sensor captures a time-series data $\mv u$ for $L^{\rm max}$ time steps, containing meaningful information from an unknown time $l^{\rm start}$ for a duration of $L^{\rm sig}$. A change detector is applied simultaneously to determine whether the sensed sequence contains a signal of interest. (\emph{ii}) \emph{Wake-up signal transmission}: If a signal of interest is detected at some specific time, the wake-up Tx and encoding SNN are turned on, and a WUS is transmitted by the wake-up Tx. (\emph{iii}) \emph{Data transmission}: After a fixed delay following the transmission of the WUS, the input signal $\{\mv u_l\}_{l=l^{\rm start}}^{L^{\rm max}}$ is processed by the encoding NPU, and the output spikes are modulated using impulse radio (IR) and transmitted over the wireless channel to the main Rx. (\emph{iv}) \emph{Wake-up signal reception and activation of the main radio}: The WUS is detected by the wake-up Rx, leading to the activation of the main Rx. (\emph{v}) \emph{Decision Making}: Upon waking up of the main Rx, the NPU at the receiver side processes the signal received by the main Rx to make an inference decision. Our goal is to optimize the threshold applied by signal detection, WUS detection, and decision making in order to provably control the average loss of the decision to a predetermined level, while minimizing the overall energy consumption.   }
	\label{model}
\end{figure} 
 
 As illustrated in Fig. 1, when deployed for \emph{split computing} in edge-based systems \cite{matsubara2020split, matsubara2020head}, remote NPUs, each carrying out part of the computation, can reduce the communication power budget by communicating asynchronously using sparse \emph{impulse radio} (IR) waveforms \cite{skatchkovsky2020end, 10016643, 9997098}, a form of  ultra-wide bandwidth (UWB) spread-spectrum signaling. After extensive research activity in the early 2000s (see, e.g., \cite{zhao2002performance}), UWB has recently  re-emerged as a prominent solution for low-power, low-range connectivity. For example, Apple has incorporated a UWB radio in the most recent iPhone models, from iPhone 11 onwards, for precision ranging \cite{apple}. Furthermore, the IEEE 802.15.4z standard includes the items Enhanced UWB Physical Layers (PHYs) and Associated Ranging Techniques which cover the reliability, accuracy, and security of UWB communications \cite{9179124}. Additionally, several high-profile academic proposals have advocated for the use of impulse radio in next-generation wireless interfaces \cite{fettweis20216g}.

Using IR waveforms, the input signal's sparsity, which depends on the semantics of the information processing task, translates directly into energy savings both in terms of computation and communication. This property has been leveraged in recent works such as \cite{he2022implantable} and \cite{lee2024asynchronous} for innovative applications including the sensing of peripheral nerves and brain-computer interfaces. Both references above present hardware validations of the concept, with \cite{lee2024asynchronous} reporting on a testbed involving 78 sensors (see also the news story\footnote{https://spectrum.ieee.org/brain-machine-interface-2667619198}). 

However, the power savings afforded by sparse transmitted signals  are limited to the transmitter's side, which can transmit impulsive waveforms only at the times of synaptic activations. The main contributor to the overall energy consumption remains the power required to maintain the main radio on \cite{7990121, jouni20221, win1998impulse}. To address this architectural problem, as seen in Fig. 1, this work proposes a novel architecture that integrates a \emph{wake-up radio} mechanism within a split computing system consisting of remote, wirelessly connected, NPUs. 

 Wake-up radios introduce a low-cost radio at the transmitter and at the receiver. The wake-up transmitter monitors the sensed signals, deciding when to transmit a \emph{wake-up signal} (WUS) to the receiver. The wake-up receiver operates at a much reduced power as compared to the main receiver radio, and its sole purpose is detecting the WUS. Upon detection of the WUS, the main radio is activated \cite{jouni20221, hoglund20243gpp, 7990121,yomo2012wake,shiraishi2020content}. 

 A key challenge in the design of a wake-up radios is the selection of thresholds for sensing and WUS detection, and decision making. A conventional solution would be to calibrate the thresholds via \emph{on-air} testing, trying out different thresholds via testing on the actual physical system. On-air calibration would be expensive in terms of spectral resources, and there is generally no guarantee that the selected thresholds would provide desirable performance levels for the end application. 

To address this design problem, as illustrated in Fig. 2, this paper proposes a novel methodology that leverages the use of a \emph{digital twin}, i.e., a simulator, of the physical system, coupled with a sequential statistical testing approach that provides theoretical reliability guarantees \cite{laufer2022efficiently, angelopoulos2021learn}.  

\subsection{Related Work}
\emph{Neuromorphic communications}: Neuromorphic communication, introduced in \cite{skatchkovsky2020end}, integrates event-driven principles from neuromorphic computing into wireless communication systems for efficient sensing, communications, and decision-making. Reference \cite{10016643} presented an architecture for wireless cognition that incorporates neuromorphic sensing, processing, and IR communications for multiple devices, leveraging time hopping for asynchronous multi-access. Motivated by the potential of IR for radar sensing \cite{bozorgi2021rf}, reference \cite{9997098} introduced a neuromorphic integrated sensing and communication system, which targets simultaneous data transmission and target detection. In \cite{dakic2023spiking}, a neuromorphic computing-based detector was implemented at a satellite receiver, whose goal was to detect Internet-of-Things signals in the presence of significant uplink interference. A hardware implementations of the system introduced in \cite{10016643} was detailed in \cite{lee2024asynchronous}, showing the potential of the approach to scale to thousands of nodes. The solution presented in \cite{lee2024asynchronous} leveraged energy harvesting.

Decentralized implementations of NPUs were studied in \cite{borsos2022resilience}, while assuming conventional digital communications. A corresponding optimal resource allocation problem was investigated in \cite{liu2024energy}.

\emph{Impulse radio for neuromorphic communications}: IR has been proposed for wireless communication of digital packets between SNN chips in \cite{cassidy2008impulse}, and for transmitting time-encoded analog signals, similar to those measured by neuromorphic sensors, for biomedical applications in \cite{shahshahani2015all}. Additionally, a combination of neuromorphic sensing, time-based computing, and IR has been utilized in \cite{peper2015low} to implement a consensus method based on device-to-device local communications for computing the maximum of scalar observations. In \cite{10016643, skatchkovsky2020end}, IR waveform were used to modulate the spiking signal for wireless transmission. 

\emph{Wake-up radio}:  Wake-up radios can reduce energy consumption in wireless communication systems by keeping the main receiver radio off until an incoming signal of interest is detected \cite{7990121}. In 3GPP Release 18, two wake-up receiver (WUR) architectures are introduced, using either a radio frequency envelope detector or an on-chip local oscillator approach \cite{hoglund20243gpp}. The first type of architecture is characterized by low complexity, low cost, and extremely low energy consumption. In contrast, the second architecture requires more complex components, like on-chip local oscillators. This results in higher energy consumption, but the benefits include better sensitivity and robustness to interferers.

For the design of WUS, two main candidates in 3GPP Release 18 are on-off keying (OOK)-based WUS and OFDM-based WUS \cite{hoglund20243gpp}. The OFDM-based signal structure does not require significant changes on the transmitter, while OOK-based WUS is an attractive choice for receivers with low complexity. 

Wake-up radios have been integrated into a number of wireless systems. For example, in \cite{pegatoquet2018wake}, a multi-access protocol was introduced that facilitates fully asynchronous communication among network devices, while reference \cite{tang2016tight} focused on WURs for wireless local area networks. Reference \cite{jouni20221} proposed a neuromorphic enhanced WUR, tailored for brain-inspired applications using OOK-modulated WUSs.

\emph{Digital twins for wireless communication}: Digital twinning is currently viewed as a promising enabling tool for the design and monitoring of next-generation wireless systems implementing machine learning modules \cite{khan2022digital}. For example, reference \cite{ruah2023bayesian}  proposed a Bayesian framework for the development of a DT platform aimed at the control, monitoring, and analysis of a multi-access communication system. The papers \cite{jiang2023digital} and \cite{morais2024localization} proposed the use of digital twinning for the design of beam prediction and localization, respectively.

\emph{Guaranteed reliability for machine learning in wireless communications}: Conformal prediction (CP) uses past experience to determine precise levels of confidence in new predictions \cite{shafer2008tutorial}. This approach guarantees that, with a specified confidence level, future predictions will fall within the prediction regions, thereby providing reliable estimates of uncertainty. For the application of CP to wireless communication, \cite{cohen2023calibrating} applied CP to the design of AI for communication systems in conjunction with both frequentist and Bayesian learning, focusing on the key tasks of demodulation, modulation classification, and channel prediction.

\emph{Learn then Test} (LTT) is a framework for the selection of hyperparameters in  pre-trained machine learning models that  satisfy finite-sample statistical guarantees \cite{angelopoulos2021learn}. Like CP and conformal risk control (CRC), it relies on the use of  calibration data, but it does not require the monotonicity assumption of CRC. As a result, it applies to more general settings, such as problems with multiple hyperparameters. Being a generic framework, LTT requires a dedicated effort to be tailored to a specific problem setting. To the best of our knowledge, ours is the first work that proposes a methodology for the application of  LTT to the design of communications system.

Regarding the comparison with artificial neural networks (ANNs) in the context of wireless communication, reference \cite{skatchkovsky2020end} has shown that split computing based on SNNs and impulse radio can outperform frame-based ANN-based solutions, thanks to the benefits of event-driven communication and processing. Reference \cite{10016643} extended these benefits to multi-device scenarios, for which impulse radio transmission can facilitate energy-efficient multi-access protocols. This was also verified experimentally by \cite{lee2024asynchronous} using a testbed involving 78 sensors, built to operate according to the principles of neuromorphic communications. Against this background, our work offers additional energy saving advantages on top of those already reported in these papers by implementing a wake-up radio receiver.

\subsection{Main Contributions}
The contribution of this paper is twofold. First, as shown in Fig.~\ref{model}, we introduce a low-power wake-up radio aided neuromorphic wireless split computing architecture, whose goal is to carry out a remote inference task in an energy efficient way. Second, we propose a novel design methodology that combines LTT with digital twinning. This methodology, dubbed DT-LTT, enhances the spectral efficiency of a direct application of LTT \cite{angelopoulos2021learn} via a digital twin-based pre-selection of candidate thresholds for sensing, detection, and decision making. The main contributions of this paper are summarized as follows.

 \emph{Architecture}: We introduce a wake-up radio aided neuromorphic wireless split computing architecture, which combines the energy savings resulting from event-driven computing at the transmitter and receiver, as well as from IR transmission, with the energy savings made possible at the receiver via the introduction of a WUR. We summarize the merits of different split computing schemes in Table I, highlighting the capacity of the architecture proposed in this work to attain low energy consumption duty cycle at both transmitter and receiver via the introduction not only of IR at the transmitter but also of a WUR at the receiver. 

 As illustrated in Fig. 1, in the proposed  architecture, the NPU at the transmitter side remains idle until a signal of interest is detected by the signal detection module. Subsequently, a WUS is transmitted by the wake-up transmitter over the channel to the wake-up receiver, which activates the main receiver. The IR transmitter modulates the encoded signals from the NPU, and sends them to the main receiver. The NPU at the receiver side then decodes the received signals and make an inference decision.

\emph{Digital twin-aided design methodology with reliability guarantees}: In order to select the thresholds used at transmitter and receiver for sensing, WUS detection, and decision making, we propose a novel design methodology that integrates the LTT framework \cite{angelopoulos2021learn} with  digital twinning. The proposed methodology, dubbed DT-LTT, is of broader interest as it can be applied to any  communication system requiring the selection of hyperparameters via on-air transmission. 

To explain, consider any setting that requires the selection of hyperparameters affecting the operation of a wireless link, here the mentioned thresholds. A direct application of LTT \cite{angelopoulos2021learn} would sequentially test candidate hyperparameters via the estimation of the target performance metrics through transmissions on the wireless channel. This way, the designer would be limited to testing a few candidate hyperparameters, given the limited availability of spectral resources. 

To reduce the spectral overhead caused by hyperparameter calibration, we propose executing LTT through digital twinning. Specifically, the digital twin is leveraged to pre-select a sequence of hyperparameters to be tested using on-air calibration via LTT.  The proposed DT-LTT calibration procedure is proved to guarantee reliability of the receiver's decisions  irrespective of the fidelity of digital twin and of the data distribution. Indeed, the fidelity of the digital twin only affects the energy consumption and the informativeness of the output produced by the calibrated system. In this regard, the proposed method also  supports the  optimization of a weighted criterion involving  energy consumption and informativeness of the receiver's decision.

 \emph{Numerical evaluations}: Extensive numerical results are provided that demonstrate the advantages of the proposed digital twin-based design approach.

  \begin{table}[t!]
      \caption{Power consumption for split computing strategies operating on temporally sparse signals, e.g., for monitoring applications.}
    \centering
    \begin{tabular}{|>{\centering\arraybackslash}m{3.6cm}|>{\centering\arraybackslash}m{2cm}|>{\centering\arraybackslash}m{2cm}|}
        \hline
        \textbf{scheme (communication/computation)} & \textbf{low transmit-power duty cycle} & \textbf{low receive-power duty cycle} \\
        \hline
        frame-based/ANNs & \ding{55} & \ding{55} \\
        \hline
        event-driven/SNNs \cite{skatchkovsky2020end, 10016643, racz2022full} & \checkmark & \ding{55} \\
        \hline
        event-driven/SNNs with wake-up radio (this work) & \checkmark & \checkmark \\
        \hline
    \end{tabular}
\end{table}

\vspace{-0.3cm}
\subsection{Organization}
The remainder of the paper is organized as follows. Section \ref{sec:System Model} presents the system model for the proposed wake-up radios assisted neuromorphic split computing system. Section \ref{section 3} describes the neuromorphic receiver processing with wake-up radio and the problem of interest, while the reliable hyperparameters optimization algorithm is proposed in Section \ref{section 4}. Experimental setting and results are described in Section \ref{exp}. Finally, Section \ref{con} concludes the paper.

\section{Background}
In this section, we provide background material that will be used in this work to introduce the proposed neuromorphic split computing system. Specifically, we first review  reliable decision-making via prediction sets and CP \cite{shafer2008tutorial}; and then we discuss hyperparameter optimization via multiple hypothesis testing \cite{angelopoulos2021learn}.

\subsection{Reliable Decision-Making via Prediction Sets}
Reliable decision-making in machine learning requires not only accurate predictions but also a quantification of the uncertainty associated with the predictions. Conventional models often provide point predictions, which, while useful, fail to convey the uncertainty inherent in the model’s decision-making process. This subsection reviews CP as a statistical method to calibrate prediction sets to ensure finite-sample coverage guarantees.

In  classification problems, a machine learning model is trained to map an input $\mv u$ into one out of a discrete set $\{1,\ldots, C\}$ of class labels. The  goal is to predict the most likely class $\hat{c}$ given a new input $\mv u$, along with a confidence score. It is well known that machine learning models, particularly with larger and potentially more accurate architectures, tend to be overconfident, offering an unreliable estimate of their uncertainty  \cite{achiam2023gpt, huang2024calibrating, vadera2020ursabench}.

CP addresses this limitation by providing a set of possible outcomes $\mathcal{C}$ that are statistically likely to contain the true class label $c$ with a specified confidence level $1-\alpha$, i.e.,
\begin{align}
    \Pr(c\in\mathcal{C}) \geq 1-\alpha. \label{setrelia}
\end{align}
CP leverages the scores $s_c$ associated by the underlying model to each class $c$. Scores $s_c$ are assumed here to be negatively oriented, i.e., they are smaller for classes on which the model is most confident. An example is given by the standard log-loss \cite{simeone2022machine}.  Given the scores $s_c$ for all classes $c\in\{1,...,C\}$, CP  constructs the predicted  set by including all classes whose score is below a threshold $\lambda^{\rm d}$ as 
\begin{align}
    \mathcal{C} = \{c\in\{1,\ldots, C\}: s_c \leq \lambda^{\rm d}\}, 
\end{align}
where the threshold $\lambda^{\rm d}$ is obtained based on a held-out calibration set. As detailed in Section \ref{designp}, in this work, we treat the threshold $\lambda^{\rm d}$ as one of the hyperparameters to be optimized by the system.

\subsection{Reliable Hyperparameter Optimization via Multiple-Hypothesis Testing}
In this subsection, we introduce LTT, a reliable hyperparameter optimization framework based on multiple-hypothesis testing. Consider a machine learning model whose operation is controlled by a  hyperparameter vector $\mv \lambda$, such as the learning rate for fine-tuning or the temperature in generative models \cite{feurer2019hyperparameter}. LTT searches through a pre-defined set of candidate hyperparameter vectors $\Lambda=\{\mv \lambda_1, \mv \lambda_2, \ldots, \mv \lambda_{|\Lambda|}\}$ to produce a subset of hyperparameters that are guaranteed to control the risk of the system. 

To elaborate, define as $R(\mv \lambda)$  a population risk measure that we wish to control, such as the probability of a classification error. LTT associates with each candidate hyperparameter $\mv \lambda_j\in\Lambda$, with $j=1,\ldots, |\Lambda|$,  the null hypothesis \begin{equation}\mathcal{H}(\mv \lambda_j): R(\mv \lambda_j)>\alpha,\end{equation} where $\alpha$ is the maximum tolerated risk. Accordingly, the null hypothesis $\mathcal{H}(\mv \lambda_j)$ posits that hyperparameter $\mv \lambda_j$ is unreliable. Rejecting this hypothesis hence entails a decision that hyperparameter $\mv \lambda_j$ is reliable, in the sense that it meets the reliability condition $R(\mv \lambda_j)\leq \alpha$. 

The goal of LTT is to identify a subset $\Lambda^{\rm rel}\subseteq \Lambda$ of hyperparameter vectors such that the condition \begin{equation} \Pr[\exists \mv \lambda \in \Lambda^{\rm rel} \text{ s.t. } R(\mv \lambda)>\alpha] \leq 1-\delta \end{equation} is satisfied for some target outage probability $\delta$. Accordingly, the identified set of hyperparameters $\Lambda^{\rm rel}$ contains no unreliable hyperparameter $\mv \lambda$ with probability at least $1-\delta$.

LTT relies on the evaluation of a p-value $p(\mv \lambda_j)$ for each null hypothesis $\mathcal{H}(\mv \lambda_j)$ \cite{rice2007mathematical}, and hence for each candidate hyperparameter $\mv \lambda_j$. To this end, an empirical estimate $\hat{R}(\mv \lambda_j)$ of the risk $R(\mv \lambda_j)$ is obtained by using existing data or real-world testing. The p-value measures the probability of obtaining an estimate at least as small as $\hat{R}(\mv \lambda_j)$ when assuming the validity of the null hypothesis $\mathcal{H}(\mv \lambda_j)$ that the hyperparameter $\mv \lambda_j$ is not reliable.  The p-values are then combined using methods for the control of the family-wise error rate (FWER) such as Bonferroni or fixed sequence testing. In this work, we will leverage fixed sequence testing, which tests hyperparameters sequentially. As further detailed in Section \ref{onairc}, the testing order is ideally selected to consider hyperparameters in order of decreasing expected reliability.

\section{System Model} \label{sec:System Model}
As shown in Fig.~\ref{model}, we consider an end-to-end neuromorphic remote inference system, in which the receiver (Rx) collects information from a device in order to carry out a semantic task, such as segmentation. 

At the device, also referred to as transmitter (Tx), the sensor monitors the environment continuously to detect the start of a signal of interest. When the Tx detects a semantically relevant signal, the wake-up Tx is turned on to transmit the WUS, and the encoding NPU is also activated to process the input signal. The output of the NPU is buffered, and subsequently modulated and transmitted by the IR Tx after a given delay. Upon detecting the WUS, the wake-up Rx activates the main Rx, which starts receiving after a given delay. The received signal is then processed by a decoding NPU, which produces a final decision.

In this way, the proposed architecture combines the energy savings resulting from event-driven computing at Tx and Rx, as well as from IR transmission, with the energy savings made possible at the Rx via the introduction of a WUR.

We observe that, throughout this study, the presence of NPUs at the transmitter and receiver is accounted for by considering neural models that are suitable for implementation on neuromorphic hardware. This is detailed in the next section, and it follows the approach adopted in most works in the field such as \cite{10016643, wu2022little}. Note that libraries such as Intel's Lava also simulate the operation of NPUs by implementing suitable spiking neural models. We leave it as future work  to present a full implementation integrating software-defined radios, neuromorphic hardware, and neuromorphic sensors (see also \cite{ke2024neuromorphic,lee2024asynchronous} for some initial work in this direction).

\subsection{Sensing Model}
We assume that the relevant discrete-time signal captured by the sensor has a duration of $L^{\rm sig}$ samples, with each sample $\mv u_l$ being a $D$-dimensional vector. The duration $L^{\rm sig}$ is assumed to be known and deterministic. The signal of interest is semantically associated with label information $c$. We assume that the labels take values in a finite discrete set, but extensions to continuous quantities are direct. Furthermore, the signal is produced by an information source after a random delay of $l^{\rm start}$ time instants. Specifically, during an initial random period of $l^{\text{start}}-1$ samples, the device observes a signal containing semantically irrelevant information, e.g., noise. The samples of the signal of interest is presented to the device starting at time $l^{\rm start}$. Subsequently, the device again records irrelevant signals. 

The sensor is active for a period of time equal to $L^{\rm max} \geq L^{\rm sig}$ samples. The choice of $L^{\rm max}$ entails a trade-off between energy consumption and probability of fully observing the signal of interest of duration $L^{\rm sig}$. 

\begin{figure*}[htp]
	\centering
	\includegraphics[width=5.3in]{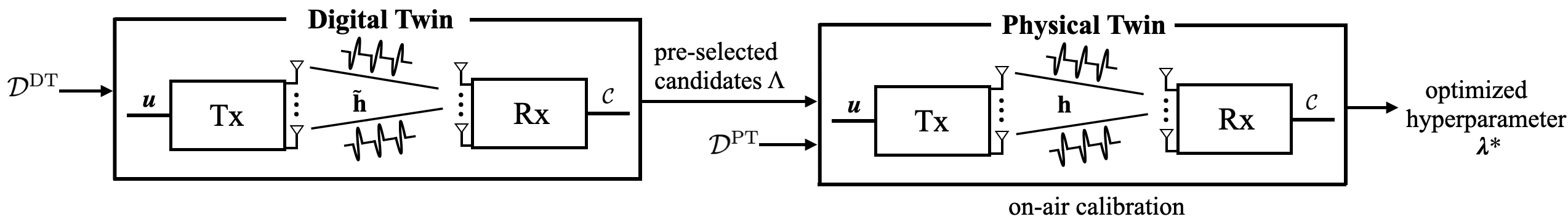}
	\caption{Hyperparameters optimization is carried out by leveraging a dataset $\mathcal{D}$ of data examples, as well as access to a simulator of the channel implemented in a digital twin. The simulator produces channel variables $\Tilde{\mathbf{h}}$ with a distribution $\Tilde{p}(\mathbf{h})$ that is generally mismatched with respect to the true distribution $p(\mathbf{h})$. In a first phase, the digital twin uses the simulator to pre-select a subset $\Lambda$ of candidate hyperparameters $\mv \lambda$. In a second phase, on-air calibration leverages transmission on the actual system (physical twin) to identify a solution $\mv \lambda^*$ that is guaranteed to satisfy the constraint in \eqref{eq:goal}.   }
	\label{dtpt}
\end{figure*} 

The sensed samples $\mv u_l$ for $l=1,2,\ldots$, are processed continuously by a \emph{signal detector} at the Tx to determine an estimate $\hat{l}^{\rm start}$ of the time $l^{\rm start}$. The signal detector updates a cumulative sum statistic $S_l$ at each time $l$ using the current sample $\mv u_l$ via an algorithm such as QUSUM \cite{qusum} or non-parametric change detection \cite{shin2022detectors}. A change is detected at time $l$ if the statistics $S_l$ exceeds a threshold $\lambda^{\rm s}$, i.e., $S_l > \lambda^{\rm s}$, and thus the wake-up Tx and encoding NPU are activated at time 
\begin{align}
    \hat{l}^{\rm start} = \min_{l\in\{1,\ldots, L^{\rm max}\}} \{S_l > \lambda^{\rm s}\}, \label{wakeup}
\end{align}
where the threshold $\lambda^{\rm s}$ is subject to optimization.

\begin{figure}[htp]
	\centering
	\includegraphics[width=3.2in]{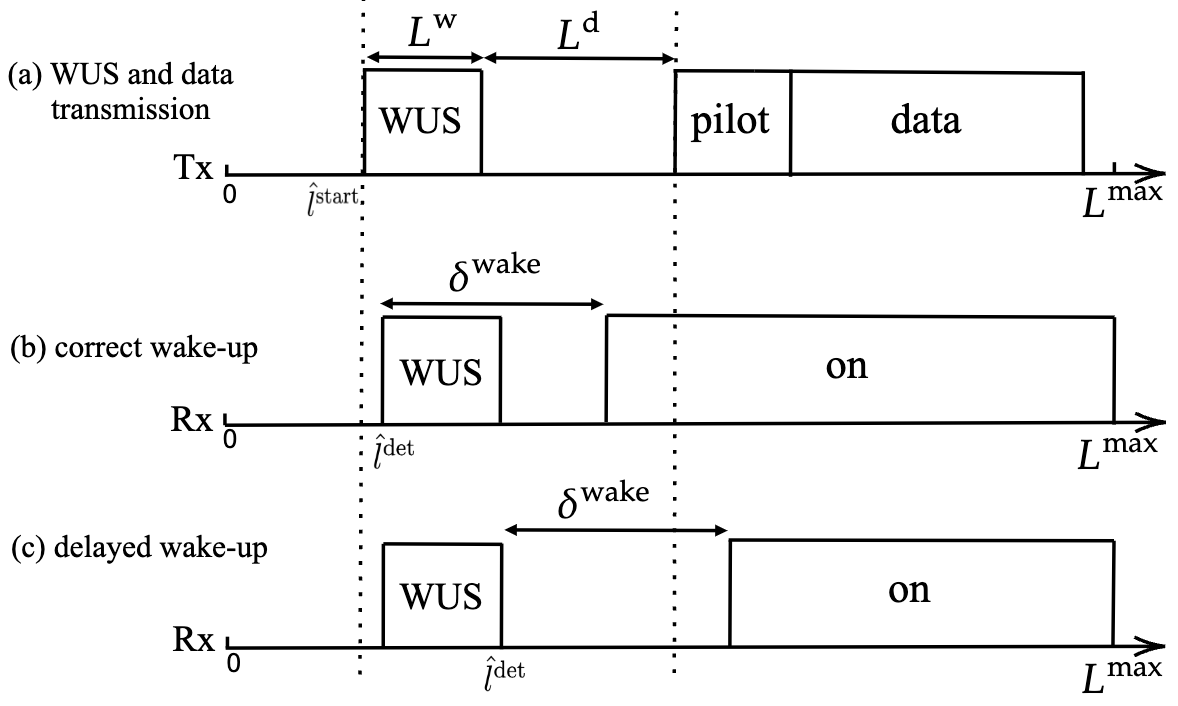}
	\caption{Illustration of the working flow of the Tx and Rx. (a) WUS and data transmission: The WUS is sent by the wake-up Tx once the signal of interest is detected at time $\hat{l}^{\rm start}$, followed by the transmission of the pilot and the data after $L^{\rm d}$ delay. (b) Correct wake-up: The wake-up receiver detects the WUS at time $\hat{l}^{\rm det}$ and activates the main receiver. The main receiver takes $\delta^{\rm wake}$ time to be fully activated. Importantly, the wake-up time of the main receiver precedes the commencement of data transmission. (c) Delayed wake-up: In this scenario, the main receiver wakes up after the data transmission has initiated, leading to data loss.    }
	\label{time}
\end{figure} 
\subsection{Neuromorphic Encoding}
Upon activation of the wake-up Tx at time $\hat{l}^{\rm start}$ in \eqref{wakeup}, an OOK-based WUS is transmitted for duration of $L^{\rm w}$ time steps. Following standard practice \cite{7990121}, as shown in Fig.~\ref{time} (top panel), data is then transmitted $L^{\rm d}$ time steps after the end of the WUS by IR Tx. The delay $L^{\rm d}$ accommodates channel delay spread, detection time of the wake-up Rx, as well as the wake-up latency of the main Rx \cite{7990121}.

The encoding NPU processes samples $\mv u_l$ starting from time $\hat{l}^{\rm start}$. For each time instant $l \in [\hat{l}^{\rm start}, L^{\rm max}]=\hat{l}^{\rm start}, \hat{l}^{\rm start}+1, \ldots, L^{\rm max}$, the encoding NPU produces an $N^{\rm T} \times 1$ vector 
\begin{align}
    \mv x_l=f_{\scalebox{0.7}{$\mv \theta^e$}}(\mv u_l) \label{spike}
\end{align}
from its $N^{\rm T}$ readout neurons. In \eqref{spike}, the vector $\mv \theta^e$ is the parameter vector of the encoding NPU. The output spiking vectors $\mv x_l$ for $l \in [\hat{l}^{\rm start}, L^{\rm max}]$ are buffered and transmitted in a first-in-first-out manner starting at time $\hat{l}^{\rm start}+L^{\rm w}+L^{\rm d}$, i.e., after the transmission of the WUS and the delay $L^{\rm d}$.

\subsection{IR Transmission Model}
The wake-up Tx is equipped with one antenna, while the IR transmitter has $N^{\rm T}$ antennas. Both transmitters adopts IR to modulate their respective transmitted signal $s_w(t)$ and $\{s_i(t)\}_{i=1}^{N^{\rm T}}$. Note that this is not a requirement for the wake-up radio, and is assumed here to facilitate a low-complexity implementation. Bandwidth expansion, leveraging time hopping (TH) \cite{win1998impulse}, is utilized to manage interference between antennas of the transmitting device during data transmission. 

Accordingly, each time step $l$ of the sensed signal $\mv u_l$ comprises $L^{\rm b} \geq 1$ chips on the radio channel, with each chip having a duration of $T_c$ seconds. Consequently, each time step $l$ spans $L^{\rm b}T_c$ seconds, and hence $L^{\rm b}$ is referred to as \emph{bandwidth expansion factor}. The bandwidth expansion factor $L^{\rm b}$ serves as a tradeoff between latency and interference mitigation. Using TH, each $i$th antenna modulates the corresponding $m$th entry of vector $\mv x_l$ in \eqref{spike} using random time shifts across the $L^{\rm b}$ chips of the $l$th time period. This introduces temporal separation to reduce interference. 

\subsubsection{WUS Transmission} To elaborate, the antenna at the wake-up Tx modulates the OOK-based WUS using IR at each time step $l\in[\hat{l}^{\rm start}, \hat{l}^{\rm start} + L^{\rm w}-1]$. The OOK-based WUS $s_{\rm w}(t)$ is defined as \cite{win1998impulse}
\begin{align}
    s_{\rm w}(t)= \sum_{j=\hat{l}^{\rm start}}^{\hat{l}^{\rm start}+L^{\rm w}-1} x^{\rm w}_j \phi(t-jL^{\rm b}T_c), \label{twus}
\end{align}
where $x^{\rm w}_j$ represents the $j$th OOK symbol in the set $\{0,1\}$, and $\phi(t)$ denotes the OOK pulse waveform with bandwidth $1/T_c$. The WUS $s_{\rm w}(t)$ is received over a multi-path fading channel impulse response $h_w(t)$ by the wake-up Rx as 
\begin{align}
    w(t)= s_{\rm w}(t) * h_w(t) + z(t), \label{wus}
\end{align}
where $*$ denotes the convolutional operation and $z(t)$ is the white Gaussian noise with noise power $N_0$. 

\subsubsection{Pilot Transmission} As shown in Fig.~\ref{time}(a), following a pre-introduced delay of $L^{\rm d}$ after the WUS transmission, the IR transmitter is activated. To facilitate the main receiver's adaptation to the frequency-selective channel conditions, the IR transmitter transmits pilots prior to the data transmission. The pilot symbols sent from the $i$th antenna have a length of $L^{\rm p}$ and are defined as
\begin{align}
    s^{\rm p}_i(t)= \sum_{j=\hat{l}^{\rm start}+L^{\rm w}+L^{\rm d}}^{\hat{l}^{\rm start}+L^{\rm w}+L^{\rm d}+L^{\rm p}-1} \phi(t-jL^{\rm b}T_c- c^{\rm p}_{j,i}T_c), \label{pilot}
\end{align}
where $c^{\rm p}_{j,i}\in\{0,1,\ldots,L^{\rm b}-1\}$ is an integer for the $j$th pilot symbol transmitted from the $i$th antenna, representing the TH position within $L^{\rm b}$ chips. The pilot is transmitted over the multi-path fading channel impulse response $h_{i,n}(t)$, and is received at the $n$th receive antenna as
\begin{align}
    v^{\rm p}_n(t)= \sum_{i=1}^{N^{\rm T}} s_i^{\rm p}(t) * h_{i,n}(t) + z_n(t), \label{rpilot}
\end{align}
where $z_n(t)$ represents the white Gaussian noise at the $n$th receive antenna.

\subsubsection{Data Transmission} Data transmission commences once all pilot symbols have been transmitted. Each $i$th antenna at the IR transmitter modulates entry $x_{l,i}$ of the vector $\mv x_l=(x_{l,1}, \ldots, x_{l,N^{\rm T}})^T$ in \eqref{spike} at time $l\in[\hat{l}^{\rm start} + L^{\rm w} + L^{\rm d}+ L^{\rm p},  \ldots, L^{\rm max}]$, into a continuous-time signal $s_{i}(t)$, e.g., using Gaussian monopulses, and TH as
\begin{align}
    s_i(t)= \sum_{j=\hat{l}^{\rm start}+L^{\rm w}+L^{\rm d}+L^{\rm p}}^{L^{\rm max}} x_{j,i} \cdot \phi(t-jL^{\rm b}T_c- c_{j,i}T_c),
\end{align}
where $c_{j,i}$ is a random integer between $0$ and $L^{\rm b}-1$, representing TH position for the $i$th antenna at the $j$th time step.

The modulated signal $s_{i}(t)$ is then transmitted over the multi-path fading channel impulse response $h_{i,n}(t)$ to the Rx, where the received signal at the $n$th receive antenna is obtained as the superposition
\begin{align}
    v_n(t)= \sum_{i=1}^{N^{\rm T}}s_{i}(t) * h_{i,n}(t) + z_n(t). \label{transmission}
\end{align}
Note that this assume the delay $L^{\rm d}$ to be longer than the channel spread to avoid interference with the WUS.

\section{Neuromorphic Receiver Processing with a Wake-Up Radio} \label{section 3}
To save energy at the Rx, instead of keeping the main radio on continuously, the proposed system incorporates an ultra low-power wake-up Rx that monitors the ambient radio frequency (RF) environment and listens for the WUS via the received signal \eqref{wus}. This approach allows the Rx to remain in a low-power state for extended periods, activating the main radio only when a WUS is detected. In this section, we start by introducing the WUS detection process operated by the wake-up Rx, and then we describe how the main Rx operates after it has been activated. Finally, we mathematically formulate the design problem of interest, which consists of minimizing the main Rx power consumption and the informativeness of the inference while guaranteeing the desired level of reliability for the decision made at the Rx.

\subsection{WUS detection}
The wake-up Rx is always on, and it applies a correlator to detect the WUS $s_{\rm w}(t)$ in \eqref{twus} from the received signal $w(t)$ in \eqref{wus} \cite{7990121}. This is done via matched filtering, i.e., by evaluating the convolution between $w(t)$ and the complex conjugate of the WUS $s^*_{\rm w}(t)$ as
\begin{align}
    d(\tau) = \int_{-\infty}^{+\infty} w(t)s^*_{\rm w}(t-\tau) dt, \label{match}
\end{align}
and by detecting the WUS at time $\tau$ if the absolute value of the matched filter output $d(\tau)$ in \eqref{match} is larger than some threshold $\lambda^w$, i.e.,
\begin{align}
    \hat{l}^{\rm det}= \min_{l\in [1,\ldots, L^{\rm max}]} \{|d(lL^{\rm b}T_c)| \geq \lambda^{\rm w}\}, \label{lambdaw}
\end{align}
with threshold $\lambda^{\rm w}$ being subject to optimization. As a result, the wake-up time of the main Rx is given by $\hat{l}^{\rm det} + \delta^{\rm wake}$,
where $\delta^{\rm wake} \leq L^{\rm d}$ denotes the time required by the main Rx to be turned on upon the reception of WUS.

The main Rx does not miss the start of the data packet (see Fig.~\ref{time}(b)) as long as we have the inequality
\begin{align}
    \hat{l}^{\rm det} + \delta^{\rm wake} \leq \hat{l}^{\rm start} + L^{\rm w} + L^{\rm d}. \label{Ld}
\end{align}
Otherwise, the wake-up Rx misses at least some of the transmitted samples (Fig.~\ref{time}(c)).

\subsection{Main Radio Processing}
The main radio is equipped with $N^{\rm R}$ antennas, and it stays idle until time $\hat{l}^{\rm det}+ \delta^{\rm wake}$. Upon waking up, the main receiver samples the received pilot signals $\{v^{\rm p}_n(t)\}_{n=1}^{N^{\rm R}}$ and the received data signals $\{v_n(t)\}_{n=1}^{N^{\rm R}}$ at each time $l$, obtaining discrete-time pilots $\mv v^{\rm p}_l =[\mv v^{\rm p}_{l,1},\ldots, \mv v^{\rm p}_{l,N^{\rm R}}]$ and discrete-time data $\mv v_l =[\mv v_{l,1},\ldots, \mv v_{l,N^{\rm R}}]$, respectively. Here, the $n$th element represents the collection of signals by the $n$th antenna for $L^{\rm b}$ chips at time $l$, i.e., $\mv v^{\rm p}_{l,n}=\{v^{\rm p}_{n}(jT_c)\}_{j\in \mathcal{I}_l}$ and $\mv v_{l,n}=\{v_{n}(jT_c)\}_{j\in \mathcal{I}_l}$, where $\mathcal{I}_l=\{(l-1)L_b+1,\ldots, lL_b\}$.

\subsubsection{Pilot Processing Via Hypernetwork} 
A hypernetwork is a type of neural network that generates the weights for another neural network, which can enhance the adaptability of the other neural network to the channel conditions \cite{10016643}. The target network in our setting is the decoding NPU. 

Provided that the main radio has woken up in time, we assume knowledge of the time of arrival of the pilots. Accordingly, we begin by collecting all the received pilot symbols as $\mv v^{\rm p}=\{\mv v_l^{\rm p}\}_{l=\hat{l}^{\rm start}+L^{\rm w}+L^{\rm d}}^{\hat{l}^{\rm start}+L^{\rm w}+L^{\rm d}+L^{\rm p}-1}$. To process the received pilot, we implement a pre-trained hypernetwork parameterized by $\mv \psi$, such as a deep neural network (DNN). This hypernetwork takes the pilot $\mv v^{\rm p}$ as input, and produces a vector $\mv \omega$ as
\begin{align}
    \mv \omega=f_{\scalebox{0.7}{$\mv \psi$}}(\mv v^{\rm p}),
\end{align} 
in which each element is a scaling factor for each neuron in the decoding NPU. Effectively, the hypernetwork subsumes the task of channel estimation by directly mapping pilots to receiver's parameters.

Specifically, the vector $\mv \omega$ is composed of $N_d$ sub-vectors as $\mv \omega =\{\mv \omega_1, \ldots, \mv \omega_{N_d}\}$, where $N_d$ is also the number of layers in the decoding NPU. Each element $\mv \omega_s$ has a length equal to the number of neurons in layer $s$ of the decoding NPU. Thus, the weight matrix $\tilde{\mv \theta}_s^d$ for layer $s$ in the decoding NPU can be adjusted by the hypernetwork as
\begin{align}
    \mv \theta_s^d= \tilde{\mv \theta}_s^d \cdot \text{diag}\{\mv \omega_s\},
\end{align}
where $\text{diag}\{\mv \omega_s\}$ is a diagonal matrix with main diagonal given by the vector $\mv \omega_s$. We collect the updated weights of the decoding NPU as $\mv \theta^d=\{\mv \theta^d_1, \ldots, \mv \theta^d_{N_d} \}$.

\subsubsection{Information Decoding} The data signal $\mv v_l$ is fed to the NPU, which produces a $C \times 1$ vector 
\begin{align}
    \mv r_{l}=f_{\scalebox{0.7}{$\mv \theta^d$}}(\mv v_{l}) 
\end{align} 
via $C$ readout neurons. At the final time $L^{\rm max}$, the output of the decoding NPU is first processed to yield a decision variable. As a typical example, the $C \times 1$ spike count vector $\bar{\mv r}$  is obtained by first summing up all output signal $\{\mv r_l\}_{l=\hat{l}^{\rm det}+\delta^{\rm wake}}^{L^{\rm max}}$ from the $C$ readout neurons as 
\begin{align}
    \bar{\mv r}=\sum_{l^{\prime}=\hat{l}^{\rm det}+\delta^{\rm wake}}^{L^{\rm max}} \mv r_{l^{\prime}}. \label{count}
\end{align}

Focusing on a classification problem, the decoding NPU applies softmax function to the spike count vector $\bar{\mv r}$ to obtain a probability vector $\mv p=[p_1,\ldots,p_C]$. A score is assigned to each class $c$ using the log-loss as $s_c=-\log(p_c)$. The final decision is constructed in the form of a \emph{decision set} that includes the classes whose scores are smaller than a given threshold $\lambda^{\rm d}$, i.e., \cite{chen2023spikecp}
\begin{align}
    \mathcal{C} = \{c: s_c \leq \lambda^{\rm d}\}. \label{set}
\end{align}
The use of a decision set supports reliable decision making, whereby the size of the decision set $\mathcal{C}$ can be determined as a function of the uncertainty of the decision \cite{vovk2022algorithmic, angelopoulos2021learn}. This way, in contrast to standard methods such as top-$k$ prediction, the size $|\mathcal{C}|$ of the set is adapted to the difficulty of the input, providing a means to control the expected loss and to quantify the uncertainty.

\subsection{Design Problem} \label{designp}
Overall, the decision vector $\mv r$ in \eqref{count} produced by the decoding NPU at the receiver depends on the fading channels and noise experienced by WUS transmission as per \eqref{wus} and by data transmission as per \eqref{transmission}. We denote collectively all noise and channel variables as $\mathbf{h}$. While the variables in vector $\mathbf{h}$ cannot be controlled, the system can tune the hyperparameters $\mv \lambda=[\lambda^{\rm s}, \lambda^{\rm w}, \lambda^{\rm d}]$, dictating the threshold $\lambda^{\rm s}$ for input signal detection at the Tx as in \eqref{wakeup}; the threshold $\lambda^{\rm w}$ for WUS detection at the wake-up Rx as in \eqref{lambdaw}; and the threshold $\lambda^{\rm d}$ for prediction \eqref{set}. 

As the predicted set $\mathcal{C}$ in \eqref{set} depends on the input data $\mv u$, the channel variables $\mathbf{h}$, and the hyperparameter vector $\mv \lambda$, we will explicitly denote it as $\mathcal{C}(\mv u, \mathbf{h}, \mv \lambda)$. To define the problem of optimizing the hyperparameters $\lambda$, we introduce a \emph{loss function} $\ell(c, \mathcal{C}(\mv u, \mathbf{h}, \mv \lambda))$ capturing the discrepancy between the true target variable $c$ and the estimate $\mathcal{C}(\mv u, \mathbf{h}, \mv \lambda)$. The corresponding \emph{expected loss} is defined as
\begin{align} 
    L(\mv \lambda)=\mathbb{E}[\ell(c, \mathcal{C}(\mv u, \mathbf{h}, \mv \lambda))], \label{risk}
\end{align}
where the expectation is taken with respect to the data distribution $p(\mv u, c)$ of the input-output pair $(\mv u, c)$, as well as over the distribution $p(\mathbf{h})$ of the channel variables $\mathbf{h}$. 

Given pre-trained encoding and decoding NPUs, we wish to find hyperparameters $\mv \lambda$ that minimize the average energy consumption $E(\mv \lambda)$ at the Rx main radio and the size of the predicted set $\mathcal{C}(\mv u, \mathbf{h}, \mv \lambda)$, while controlling the expected loss $L(\mv \lambda)$ at some predetermined level $\alpha \in[0,1]$. Note that the focus on energy consumption of the main radio at the Rx is justified by the fact that it is typically the most significant contributor to the overall energy expenditure at the Rx \cite{jouni20221}.

The \emph{average energy} $E(\mv \lambda)$ consumed by the Rx main radio is evaluated as
\begin{align}
    E(\mv \lambda) = P^{\rm on}(L^{\rm max}-\mathbb{E}[\hat{l}^{\rm det}(\mv u, \mathbf{h}, \mv \lambda)]- \delta^{\rm wake}+1), \label{energy}
\end{align}
with $P^{\rm on}$ being the per-time-step energy consumed by the main radio when it is on, and the expectation is computed with respect to the data distribution of the input $\mv u$ and the distribution of vector $\mathbf{h}$. In fact, as illustrated in Fig.~\ref{time}, the Rx main radio is on for $L^{\rm max} -\hat{l}^{\rm det}(\mv u, \mathbf{h}, \mv \lambda)- \delta^{\rm wake}+1$. The notation $\hat{l}^{\rm det}(\mv u, \mathbf{h}, \mv \lambda)$ is introduced in \eqref{energy} to highlight the dependence of the detection time $\hat{l}^{\rm det}$ on input $\mv u$, channel $\mathbf{h}$, and hyperparameter $\mv \lambda$.

A smaller energy consumption \eqref{energy}  can be obtained by waking up the main radio later, i.e., by maximizing the expected value $\mathbb{E}[\hat{l}^{\rm det}(\mv u, \mathbf{h}, \mv \lambda)]$, but this generally comes at the cost of an increased average loss $L(\mv \lambda)$. To assess the informativeness of the predicted set $\mathcal{C}(\mv u, \mathbf{h}, \mv \lambda)$, we evaluate the \emph{average set size} as 
\begin{align}
    I(\mv \lambda) = \mathbb{E}[|\mathcal{C}(\mv u, \mathbf{h}, \mv \lambda)|],
\end{align}
where the expectation is taken with respect to the data distribution of the input $\mv u$ and the distribution of vector $\mathbf{h}$.

Overall, the design problem of interest is formulated as the constrained minimization
\begin{align} 
    &~ \underset{\mv \lambda}{\text{minimize}} ~ E(\mv \lambda) + \gamma I(\mv \lambda) \notag \\
    &~ \text{subject to}~  L(\mv \lambda) \leq \alpha,   \label{problem}
\end{align}
where $\gamma \geq 0$ is a weight factor determining the relative priority between the energy consumption $E(\mv \lambda)$ and the set size $I(\mv \lambda)$, while the parameter $\alpha >0$ specifies the desired reliability level, with a smaller $\alpha$ indicating a stricter reliability requirement. Regarding the choice of parameter $\gamma$ in \eqref{problem}, note that there is generally a tension between energy $E(\mv \lambda)$,  and set size $I(\mv \lambda)$. In fact, reducing the set size $I(\mv \lambda)$, while maintaining the desired target reliability $\alpha$, generally requires a larger energy expenditure $E(\mv \lambda)$.

\section{DT-LTT: Hyperparameters Optimization For Energy-Efficient Risk Control} \label{section 4}
As discussed in the last section, the goal of this work is to introduce a methodology for the selection of hyperparameters $\mv \lambda$ by addressing problem \eqref{problem}. In this section, we describe the proposed solution based on digital twinning and LTT \cite{angelopoulos2021learn}, a method recently introduced in statistics.

\subsection{Digital Twin-based Optimization}
Addressing problem \eqref{problem} is made complicated by the fact that we do not assume knowledge of the distribution $p(\mv u, c)$ of each data pair $(\mv u, c)$, consisting of sensed signal $\mv u$ and label $c$, and we also do not have access to the distribution $p(\mathbf{h})$ of the channel variables $\mathbf{h}$. To obtain information about the data distribution $p(\mv u, c)$, we make the common assumption that a dataset $\mathcal{D}=\{(\mv u_n, c_n)\}_{n=1}^{|\mathcal{D}|}$ is available, where each pair $(\mv u_n, c_n)$ of signal $\mv u_n$ and label $c_n$ is generated in an independent and identically distributed (i.i.d.) manner from the distribution $p(\mv u, c)$. Note that each pair is thus produced under an independent channel realization 
from distribution $p(\mathbf{h})$. Furthermore, to facilitate the collection of information about the distribution $p(\mathbf{h})$ of the channel variables, we assume access to a simulator in a \emph{digital twin} of the system. As illustrated in Fig.~\ref{dtpt}, the simulator can produce samples $\Tilde{\mathbf{h}}$ from a distribution $\Tilde{p}(\mathbf{h})$ that is generally different from the true distribution $p(\mathbf{h})$. The \emph{fidelity} of the simulator depends on how similar the distribution $p(\mathbf{h})$ and $\Tilde{p}(\mathbf{h})$ are.

With this information, DT-LTT aims at solving a relaxation of problem \eqref{problem}, in which the constraint is required to be satisfied with a user-determined probability $1-\delta$ with $\delta \in(0,1)$. The resulting problem is defined as 
\begin{align}
    &~ \underset{\lambda}{\text{minimize}} ~ E(\mv \lambda) + \gamma I(\mv \lambda) \notag \\
    &~ \text{subject to}~  \Pr\big[L(\mv \lambda) \leq \alpha\big] \geq 1-\delta,   \label{eq:goal}
\end{align}
where the probability $\Pr[\cdot]$ is taken with respect to the random realization of the dataset $\mathcal{D}$ and the channel $\mathbf{h}$. Note that the probability in \eqref{eq:goal} cannot be evaluated given that the distribution $p(\mv u, c)$ and $p(\mathbf{h})$ are unknown.

\subsection{Digital Twin-Based Pre-Selection of Candidate Solutions}

In order to address problem \eqref{eq:goal}, we follow a two-stage approach illustrated in Fig.~\ref{dtpt}. In the first phase, the digital twin pre-selects a subset $\Lambda$ of candidate hyperparameter vectors $\mv \lambda$. The pre-selected candidates in set $\Lambda$ are then tested in the following phase of \emph{on-air calibration} to identify a hyperparameter vector $\mv \lambda^*$ that provably satisfies the constraint in \eqref{eq:goal}. Reducing the size of the candidate solutions via the use of the digital twin supports a more efficient use of the physical channel resources during on-air calibration, as fewer options need to be evaluated using transmission on the wireless channel.

At a technical level, as detailed in the Appendix, the proposed approach leverages the freedom in the LTT scheme reviewd in Section II to choose  any fixed sequence of hyperparameter vectors for testing of the reliability condition (\ref{eq:goal}). Our proposed method, DT-LTT, determines the sequence of hyperparameter vectors by leveraging a digital twin model.

To start, the dataset $\mathcal{D}$ is randomly partitioned into two subsets, namely the dataset $\mathcal{D}^{\rm DT}$ to be used with the simulator produced by the digital twin and the dataset $\mathcal{D}^{\rm PT}$ to be leveraged for on-air calibration in the physical system. To carry out the pre-selection of a subset $\Lambda$ of hyperparameter, the digital twin addresses the \emph{multi-objective problem} 
\begin{align}
    \underset{\mv \lambda}{\text{minimize}}~ \{ \hat{L}^{\rm DT}(\mv \lambda), \hat{E}^{\rm DT}(\mv \lambda) + \gamma \hat{I}^{\rm DT}(\mv \lambda)\}, \label{mutio}
\end{align}
where the objectives $\hat{L}^{\rm DT}(\mv \lambda)$, $\hat{E}^{\rm DT}(\mv \lambda)$ and $\hat{I}^{\rm DT}(\mv \lambda)$ are empirical estimates obtained at the digital twin for the expected loss \cite{simeone2022machine}
\begin{align}
     \hat{L}^{\rm DT}(\mv \lambda)=\frac{1}{|\mathcal{D}^{\rm DT}|} \sum_{n=1}^{|\mathcal{D}^{\rm DT}|} \ell\big(c, \mathcal{C}(\mv u_n, \Tilde{\mathbf{h}}_n, \mv \lambda)\big), \label{erisk}
\end{align}
the average energy consumption
\begin{align}
    \hat{E}^{\rm DT}(\mv \lambda)= P^{\rm on}\bigg(L^{\rm max}-\frac{1}{|\mathcal{D}^{\rm DT}|} \sum_{n=1}^{|\mathcal{D}^{\rm DT}|} \hat{l}^{\rm det}(\mv u_n, \Tilde{\mathbf{h}}_n, \mv \lambda) - \delta^{\rm wake} +1\bigg), \label{epower}
\end{align}
and the average set size 
\begin{align}
    \hat{I}^{\rm DT}(\mv \lambda) = \frac{1}{|\mathcal{D}^{\rm DT}|} \sum_{n=1}^{|\mathcal{D}^{\rm DT}|}|\mathcal{C}(\mv u_n, \Tilde{\mathbf{h}}_n, \mv \lambda)|. \label{esize}
\end{align}

The empirical estimates \eqref{erisk}, \eqref{epower} and \eqref{esize} are obtained by using the dataset $\mathcal{D^{\rm DT}}$ and transmission simulated using channels $\Tilde{\mathbf{h}}_n \sim \Tilde{p}(\mathbf{h})$ generated by digital twin. As shown in Fig.~\ref{dt}, the digital twin uses an arbitrary multi-objective optimization algorithm to identify a discrete subset $\Lambda$ of values of the hyperparameter $\mv \lambda$ such that the resulting estimates $\big(\hat{L}^{\rm DT}(\mv \lambda), \hat{E}^{\rm DT}(\mv \lambda)+ \gamma \hat{I}^{\rm DT}(\mv \lambda)\big)$ lie on the Pareto front of the set of achievable values for the pair $\big(\hat{L}^{\rm DT}(\mv \lambda), \hat{E}^{\rm DT}(\mv \lambda)+ \gamma \hat{I}^{\rm DT}(\mv \lambda)\big)$. Mathematically, each vector $\mv \lambda$ included in the candidate set $\Lambda$ satisfies the condition
\begin{align}
    \nexists \mv \lambda^{\prime} & ~\text{such that}~ \hat{L}^{\rm DT}(\mv \lambda^{\prime}) < \hat{L}^{\rm DT}(\mv \lambda) ~\text{and}~ \notag \\
    &\hat{E}^{\rm DT}(\mv \lambda^{\prime})+ \gamma \hat{I}^{\rm DT}(\mv \lambda^{\prime}) < \hat{E}^{\rm DT}(\mv \lambda)+ \gamma \hat{I}^{\rm DT}(\mv \lambda)
\end{align}
that no other hyperparameter $\mv \lambda^{\prime}$ improves both empirical loss and empirical energy consumption plus the weighted set size.

\begin{figure*}[htp]
	\centering
	\includegraphics[width=5.3in]{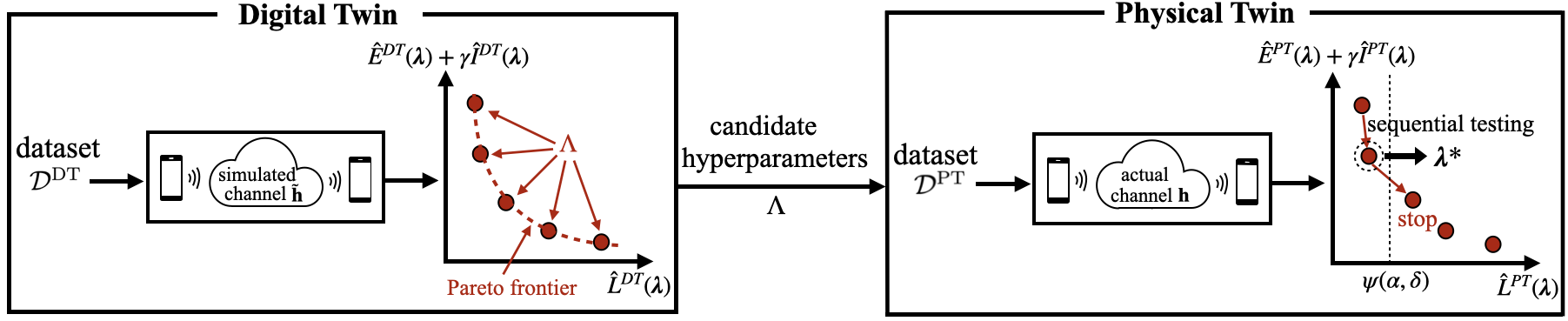}
	\caption{Illustration of the proposed DT-LTT design strategy: During the first phase of pre-selection, the digital twin determines a subset of candidate hyperparameters that yield estimated $\hat{E}^{\rm DT}(\mv \lambda)+\gamma \hat{I}^{\rm DT}(\mv \lambda)$ and loss $\hat{L}^{\rm DT}(\mv \lambda)$ on the Pareto frontier. Then, during on-air calibration, the physical twin transmits on the actual channel to test the candidates in set $\Lambda$ sequentially, stopping when the estimated loss crosses a threshold $\psi(\alpha, \delta)$. The solution $\mv \lambda^*$ is then obtained by choosing the value of $\mv \lambda$ that yields the minimum estimated objective $\hat{E}^{\rm PT}(\mv \lambda)+\gamma \hat{I}^{\rm PT}(\mv \lambda)$, while guaranteeing the inequality $\hat{L}^{\rm PT}(\mv \lambda)< \psi(\alpha, \delta)$.  }
	\label{dt}
\end{figure*} 

\subsection{On-Air Calibration} \label{onairc}
Given the pre-selected candidate solutions in set $\Lambda$, on-air calibration aims at selecting a value $\mv \lambda$ that approximately solves problem \eqref{mutio}, ensuring the validity of the reliability constraint in \eqref{eq:goal}. To this end, the solutions in set $\Lambda$ are first ordered with respect to the loss value $\hat{L}^{\rm DT}(\mv \lambda)$ in \eqref{erisk} as
\begin{align}
    \hat{L}^{\rm DT}(\mv \lambda_1) \leq \hat{L}^{\rm DT}(\mv \lambda_2) \leq \ldots \leq \hat{L}^{\rm DT}(\mv \lambda_{|\Lambda|}).
\end{align}
On-air calibration evaluates the solutions in set $\Lambda$ in the order $\mv \lambda_1, \mv \lambda_2, \ldots$, selecting a value $\mv \lambda^{*}$ that is guaranteed to satisfy constraint \eqref{eq:goal}, while reducing as much as possible the weighted sum of energy consumption and set size.

For any hyperparameter $\mv \lambda_j$ being tested, using transmission on the actual physical channel, the physical twin evaluates empirical expected loss
\begin{align}
     \hat{L}^{\rm PT}(\mv \lambda_j)=\frac{1}{|\mathcal{D}^{\rm PT}|} \sum_{n=1}^{|\mathcal{D}^{\rm PT}|} \ell(c, \mathcal{C}(\mv u_n, \mathbf{h}_n, \mv \lambda_j)), \label{crisk}
\end{align}
the empirical energy consumption 
\begin{align}
    \hat{E}^{\rm PT}(\mv \lambda_j)=&P^{\rm on}\bigg(L^{\rm max}-\frac{1}{|\mathcal{D}^{\rm PT}|} \sum_{n=1}^{|\mathcal{D}^{\rm PT}|} \hat{l}^{\rm det}(\mv u_n, \mathbf{h}_n, \mv \lambda_j) \notag \\
    &- \delta^{\rm wake} +1\bigg) \label{cpower}
\end{align}
and the empirical set size 
\begin{align}
    \hat{I}^{\rm PT}(\mv \lambda_j) = \frac{1}{|\mathcal{D}^{\rm PT}|} \sum_{n=1}^{|\mathcal{D}^{\rm PT}|}|\mathcal{C}(\mv u_n, \mathbf{h}_n, \mv \lambda_j)| \label{csize}
\end{align}
by transmitting on actual channel realizations $\mathbf{h}_n \sim p(\mathbf{h})$. Note that the channel realization  $\mathbf{h}_n$ is not known and not required to evaluate the estimates \eqref{crisk}, \eqref{cpower} and \eqref{csize}. 
The estimates  \eqref{crisk}, \eqref{cpower} and \eqref{csize} are evaluated successively for the candidate solutions $\mv \lambda_1,\mv \lambda_2, \ldots$, until a stopping criterion is satisfied.

Specifically, as illustrated in Fig.~\ref{dt}, the evaluation of candidate solutions $\mv \lambda_1, \mv \lambda_2,\ldots$ stops at the first value $j^{\rm stop}$ for which the estimated loss $\hat{L}^{\rm PT}(\mv \lambda_{j^{\rm stop}})$ in \eqref{crisk} exceeds the threshold 
\begin{align}
    \psi(\alpha, \delta) = \alpha - \sqrt{\frac{- \ln (\delta)}{2|\mathcal{D}^{\rm PT}|}}, \label{threshold}
\end{align}
which is a function of the dataset size $|\mathcal{D}^{\rm PT}|$, of the target expected loss $\alpha$ in \eqref{eq:goal}, and of the probability bound $\delta$ in \eqref{eq:goal}. For the optimal hyperparameter $\mv \lambda^*$ to be well defined, one needs to ensure the condition
\begin{align}
    \hat{L}^{\rm PT}(\mv \lambda_1) < \psi(\alpha, \delta).  \label{feasible}
\end{align} 
If condition \eqref{feasible} is not met, the decoding NPU makes a \emph{secure} decision by including all classes in the predicted set $\mathcal{C}$ in \eqref{set}, while saving energy by keeping the main receiver off. This amounts to the choice $\mv \lambda^*=[\lambda^{\rm s}=\infty, \lambda^{\rm w}= \infty, \lambda^{\rm d}=\infty]^T$.

\begin{algorithm}[t]
  \caption{Digital Twin-Based Learn-then-Test (DT-LTT) Calibration}\label{rca}
  \begin{algorithmic}[1]
    \STATE {\textbf{Initialization:} Dataset $\mathcal{D}^{\rm DT}$, dataset $\mathcal{D}^{\rm PT}$,  risk tolerance $\alpha\in[0,1]$, and error level $\delta \in[0,1]$} \\
    \underline{\emph{Digital Twin-based Pre-selection of Candidate Solutions}}:\\
    \STATE Using the simulated channel $\Tilde{\mathbf{h}} \sim p(\Tilde{\mathbf{h}})$, identify a subset $\Lambda$ of the candidate solutions $\mv \lambda$ such that each $\mv \lambda\in\Lambda$ returns estimates $\big(\hat{L}^{\rm DT}(\mv \lambda), \hat{E}^{\rm DT}(\mv \lambda)+ \gamma \hat{I}^{\rm DT}(\mv \lambda) \big)$ in \eqref{erisk}, \eqref{epower} and \eqref{esize} on the Pareto frontier.    \\
    \underline{\emph{On-Air Calibration}}:\\
    \STATE Order the solutions in set $\Lambda$ as $\hat{L}^{\rm DT}(\mv \lambda_1) \leq \hat{L}^{\rm DT}(\mv \lambda_2) \leq \ldots \leq \hat{L}^{\rm DT}(\mv \lambda_{|\Lambda|})$. \\
    \FOR{$j=1, 2,\ldots, |\Lambda|$}
        \STATE Estimate expected loss $\hat{L}^{\rm PT}(\mv \lambda_j)$, energy consumption $\hat{E}^{\rm PT}(\mv \lambda_j)$ and set size $\hat{I}^{\rm PT}(\mv \lambda_j)$ in \eqref{crisk}, \eqref{cpower} and \eqref{csize} using the actual channel $\mathbf{h}\sim p(\mathbf{h})$. \\
        \STATE \textbf{if}~ $j=1$~\text{and}~$\hat{L}^{\rm PT}(\mv \lambda_j) >\psi(\alpha, \delta)$
        \STATE ~~~~\text{Set} $\mv \lambda^*=[\lambda^{\rm s}=\infty, \lambda^{\rm w}= \infty, \lambda^{\rm d}=\infty]^T$ (secure solution).
        \STATE \textbf{else if}~ $j>1$~\text{and}~$\hat{L}^{\rm PT}(\mv \lambda_j) >\psi(\alpha, \delta)$
        \STATE~~~~\text{Set} $\mv \lambda^*$ using \eqref{riskltt}.
        \STATE \textbf{end if}
    \ENDFOR
  \end{algorithmic}
\end{algorithm}

Assuming that such value exists, finally, the selected value $\mv \lambda^{*}$ is obtained by choosing the value $\mv \lambda_j$ with $j\in\{1,\ldots, j^{\rm stop}\}$ that returns the smallest estimated sum $\hat{E}^{\rm PT}(\mv \lambda_j)+\gamma \hat{I}^{\rm PT}(\mv \lambda_j)$, i.e.,
\begin{align}
\mv \lambda^{*}= \mv \lambda_{j^{*}}, ~\text{with}~ j^*=\argmin_{j\in\{1,\ldots, j^{\rm stop}\}} \{\hat{E}^{\rm PT}(\mv \lambda_j) +\gamma \hat{I}^{\rm PT}(\mv \lambda_j)\}.
\label{riskltt}
\end{align} 

The overall proposed calibration procedure is described in Algorithm \ref{rca}. As proved next, by the properties of LTT \cite{angelopoulos2021learn}, DT-LTT  guarantees the constraint \eqref{eq:goal} irrespective of the true, unknown, distributions $p(\mv u, c)$ and $p(\mathbf{h})$, and irrespective of the fidelity of the digital twin.

\begin{theorem}[\textbf{Reliability of DT-LTT}]
\label{theor}
By setting the hyperparameter vector $\mv \lambda^*$ as in Algorithm \ref{rca}, DT-LTT satisfies the inequality 
\begin{align}
    \Pr[L(\mv \lambda^*) \leq \alpha] \geq 1-\delta \label{theo}
\end{align}
holds for any realizations of dataset $\mathcal{D}^{\rm DT}$, simulated channels $\{\Tilde{\mathbf{h}}_n \sim \Tilde{p}(\mathbf{h})\}_{n=1}^{|\mathcal{D}^{\rm DT}|}$, with probability in \eqref{theo} evaluated with respect to the randomness of the dataset $\mathcal{D}^{\rm PT}$ and the true channels  $\{\mathbf{h}_n \sim p(\mathbf{h})\}_{n=1}^{|\mathcal{D}^{\rm PT}|}$.
\end{theorem}
\begin{proof}
    The proof is provided in the Appendix.
\end{proof}

\section{Experiments} \label{exp}
In this section, we present numerical results that validate the proposed design and analysis.
\subsection{Setting}
To test the proposed DT-LTT calibration method, we consider a neuromorphic wireless communication link over a multi-path fading channel, whose goal is to support reliable image classification at the receiver. The transmitter is equipped with $N^{\rm T}=10$ antennas, each modulating the spiking signal produced by the corresponding neuron of the encoding NPU, while the receiver has $N^{\rm R}=2$ antennas.  All antennas share the same multipath delays, with delay of the $i$th path equal to the $i$th chip time. The signal-to-noise ratio (SNR) per time step is defined as the ratio of the transmission power, which is assumed to be the same for WUS, pilots, and data transmission, over the noise power. We set the SNR to 10 dB.

As in \cite{10016643}, the encoding NPU is a fully-connected SNN featuring one hidden layer comprising 600 neurons and an output layer with 10 neurons, while the decoding NPU is designed as an SNN with a single hidden layer containing 200 neurons and an output layer consisting of 10 neurons, each representing one of the 10 classes. The hypernetwork is implemented as an ANN with two hidden layers, containing 800 and 500 neurons, respectively.

Unless stated otherwise, the maximum observation period for each data $\mv u$ is $L^{\rm max}=60$ time steps, with the duration for the signal of interest fixed at $L^{\rm sig}=40$. During this period, we repetitively present an input image to be classified for 40 time steps. The initial time $l^{\rm start}$ is determined by drawing from a discrete uniform distribution in the set $\{1, L^{\rm max} -L^{\rm sig}\}$. Subsequently, the initial $l^{\rm start}$ and the last $L^{\rm max}-L^{\rm sig}-l^{\rm start}$ time samples of $\mv u$ are generated independently using a Bernoulli distribution. 

To implement the QUSUM algorithm, the irrelevant signals are modelled as Bernoulli i.i.d. samples with probability $p^{\rm noise}$, while relevant signals are also modelled as Bernoulli i.i.d. variables with a spiking probability $p^{\rm sig}$ estimated from the training data.

For IR transmission, the duration of the WUS is set to $L^{\rm w}=2$, and the duration for the pilot is also set to $L^{\rm p}=2$. The delay added by the transmitter is $L^{\rm d}=3$ time steps, and the wake-up time $\delta^{\rm wake}=2$. The power for keeping the main radio on is set to a normalized value $P^{\rm on}=1$. 

Decision are made via set prediction as in \eqref{set}, and the loss function $\ell(c, \mathcal{C})$ is a 0-1 loss that indicates whether the true label $c$ is included in the predicted set $\mathcal{C}$ or not, i.e., $\ell(c, \mathcal{C})=\mathbbm{1}(c \notin \mathcal{C})$, where $\mathbbm{1}(\cdot)$ is an indicator function. Accordingly, the average loss represents the \emph{probability of miscoverage} for the decision set $\mathcal{C}$. To evaluate the \emph{informativeness} of the set prediction, we also compute the normalized average set size $|\mathcal{C}|/C$ of the prediction set \cite{zecchin2024generalization}.

Since our focus is on the optimization of the thresholds, rather than on training, we adopt \emph{pre-trained} SNNs. Pre-training, testing, and calibration use the N-MNIST dataset, a neuromorphic dataset that comprises 60,000 training samples and 10,000 test samples. Each sample in the dataset represents a handwritten digit ranging from 0 to 9, and is presented as a $34\times 34$ pixel image. We partition the training dataset by drawing $6,000$ samples for the dataset $\mathcal{D}^{\rm DT}$ and $6,000$ samples for the dataset $\mathcal{D}^{\rm PT}$, with the remaining data points used for pre-training. Pre-training is done in an end-to-end manner without considering the wake-up radio as in \cite{10016643}. 

\begin{figure*}[htp]
	\centering
	\includegraphics[width=4.7in]{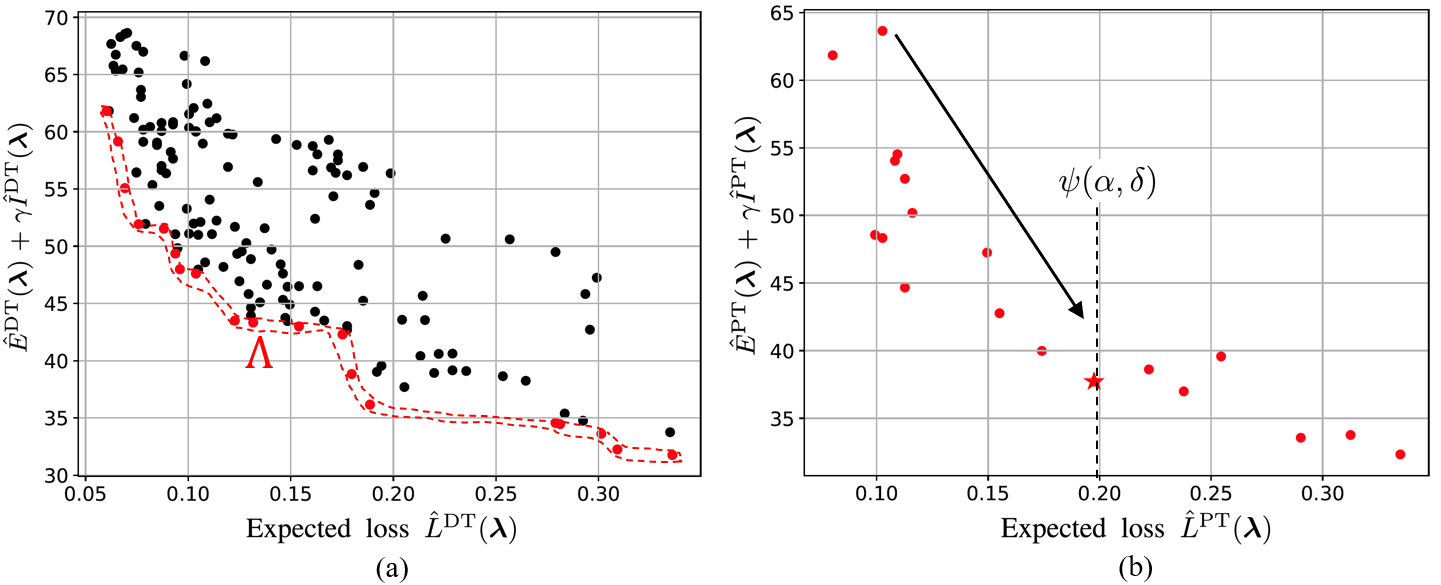}
	\caption{Illustration of the operation of DT-LTT: (a) \emph{Digital twin-based pre-selection}: Expected loss $\hat{L}^{\rm DT}(\mv \lambda)$ versus weighted sum $\hat{E}^{\rm DT}(\mv \lambda)+\gamma \hat{I}^{\rm DT}(\mv \lambda)$ estimated at the digital twin using dataset $\mathcal{D}^{\rm DT}$ and channel simulators via \eqref{erisk}, \eqref{epower} and \eqref{esize}, with each point corresponding to the evaluation of a hyperparameter $\mv \lambda$ in a grid of options. The red points represent the selected candidates, which lie on the Pareto frontier $\Lambda$. (b) \emph{On-air calibration}: Expected loss $\hat{L}^{\rm PT}(\mv \lambda)$ versus weighted sum $\hat{E}^{\rm PT}(\mv \lambda)+\gamma \hat{I}^{\rm PT}(\mv \lambda)$ estimated using actual wireless transmissions with each point representing the evaluation for one of the hyperparameters $\mv \lambda$ in the set $\Lambda$. The star is the hyperparameter selected by on-air calibration with $\alpha=0.2$, $\delta=0.05$, $\gamma=10$ and $L^{\rm max}=60$.}
	\label{paretof}
\end{figure*} 

\subsection{Benchmarks}
For comparison, we consider the following benchmarks. For all the schemes using LTT, the grid contains all threshold tuples $(\lambda^{\rm s}, \lambda^{\rm w}, \lambda^{\rm d})$ with $\lambda^{\rm s}\in \{0,1,\ldots,4\}$, $\lambda^{\rm w}\in\{0.1, 0.2, \ldots, 0.6\}$, and $\lambda^{\rm d}\in\{1,3,\ldots,9\}$.
\begin{itemize}
    \item \emph{Conventional neuromorphic wireless communications:} The conventional system is designed without signal detection and wake-up radio modules, which amounts to setting the corresponding thresholds as  $\lambda^{\rm s}=0$ and $\lambda^{\rm w}=0$. With this conventional setup, the NPUs are continuously on. Furthermore, rather than relying on the proposed adaptive set prediction strategy, in this conventional strategy, the NPU at the receiver side applies top-2 prediction to generate a prediction set, which is constructed by including the top two predicted classes with the highest spike count in the output vector \eqref{count}.
      \item \emph{LTT:} To evaluate the performance of a basic version of the  LTT algorithm, we consider a scheme that implements LTT without the use of digital twinning. This approach follows Algorithm 1, with two caveats: (\emph{i})  the step 1 of pre-selection via a digital twin is not carried out; and (\emph{ii}) the number of on-air calibration transmissions, i.e., the number of iterations of the for cycle in line 4 of Algorithm 1, is limited by the average number of Pareto points in set $\Lambda$ used by the proposed DT-LTT scheme. This way, the use of spectral resources for calibration is not increased as compared to DT-LTT. Note that this modification violates the assumptions in Theorem 1, and thus this scheme may not satisfy the reliability condition  (\ref{theo}). This approach uses a fixed test sequence within the mentioned grid of hyperparameters considering first all option with the highest threshold, and then exploring other options decreasing first $\lambda^{\rm s}\in \{0,1,\ldots,4\}$, then $\lambda^{\rm w}\in\{0.1, 0.3\}$, and finally $\lambda^{\rm d}\in\{1,5,9\}$.
    \item \emph{DT-LTT with an always-on main radio:} We also consider an \emph{always-on} variant of DT-LTT, which keeps the main receiver radio on for all time instants. In this case, the hyperparameter vector $\mv \lambda$ to be optimized contains only the threshold $\lambda^{\rm s}$ for signal detection and the threshold $\lambda^{\rm d}$ for set prediction. As for LTT, we limit the number of on-air calibration rounds to be at most equal to the number of Pareto points in set $\Lambda$ of DT-LTT.  Furthermore, we set  $\lambda^{\rm w}=0$.  Note that, for this strategy, the resulting calibration output does not depend on the parameter $\gamma$, since the energy consumption at the receiver is constant, irrespective of the selected hyperparameters  $\lambda^{\rm s}$ and $\lambda^{\rm d}$.
  
\end{itemize}

\subsection{High-fidelity Digital Twin}
We first consider a scenario in which the digital twin implements an accurate model of the channel so that the simulated channel $\Tilde{\mathbf{h}}$ follows the same distribution $p(\mathbf{h})$ as the true channel $\mathbf{h}$. For both simulated and real channels, we adopt here the standard 3GPP TR 38.901 channel model generated by Sionna, an open-source library for simulating the physical layer of wireless communication systems \cite{sionna}. We use a tapped delay line channel model from the 3GPP TR38901 specification with six paths.

To illustrate the operation of DT-LTT, Fig.~\ref{paretof}(a) presents as black and red dots the expected loss and the energy consumption plus the weighted set size estimated by the digital twin via \eqref{erisk}, \eqref{epower} and \eqref{esize} for a given realization of dataset $\mathcal{D}^{\rm DT}$ and realization of the simulated channels, when the hyperparameters $\mv \lambda$ are chosen within the mentioned grid of values.  

As seen in the figure, the expected loss and energy consumption plus weighted set size are conflicting objectives, since no hyperparameter vector $\mv \lambda$ exists that yields simultaneously the smallest loss and the smallest energy or the smallest set size. The Pareto optimal points, within the set of chosen options, are depicted as red points, constituting the set $\Lambda$ of candidates produced by the digital twin. During on-air calibration, the candidates in set $\Lambda$ are further evaluated in order of the value of the loss estimated at the digital twin.

To elaborate, in Fig.~\ref{paretof}(b), we show weighed sum of energy consumption and set size estimated during on-air calibration using one realization of the dataset $\mathcal{D}^{\rm PT}$ and channel transmissions for hyperparameters within the set $\Lambda$. As detailed in Algorithm \ref{rca}, the on-air calibration estimates the loss, energy and set size using \eqref{crisk}, \eqref{cpower} and \eqref{csize}, starting from the candidate yielding the smallest value of the loss estimated at the digital twin, and stopping once the loss estimated on the physical system exceeds the threshold $\psi(\alpha, \delta)$. Here we set $\alpha=0.2$ and $\delta=0.05$. The final solution selected by the PT is represented by the star. Note that the PT does not need to evaluate hyperparameters that result in an expected loss larger than the threshold $\psi(\alpha, \delta)$.

\begin{figure*}[htp]
	\centering
	\includegraphics[width=4.7in]{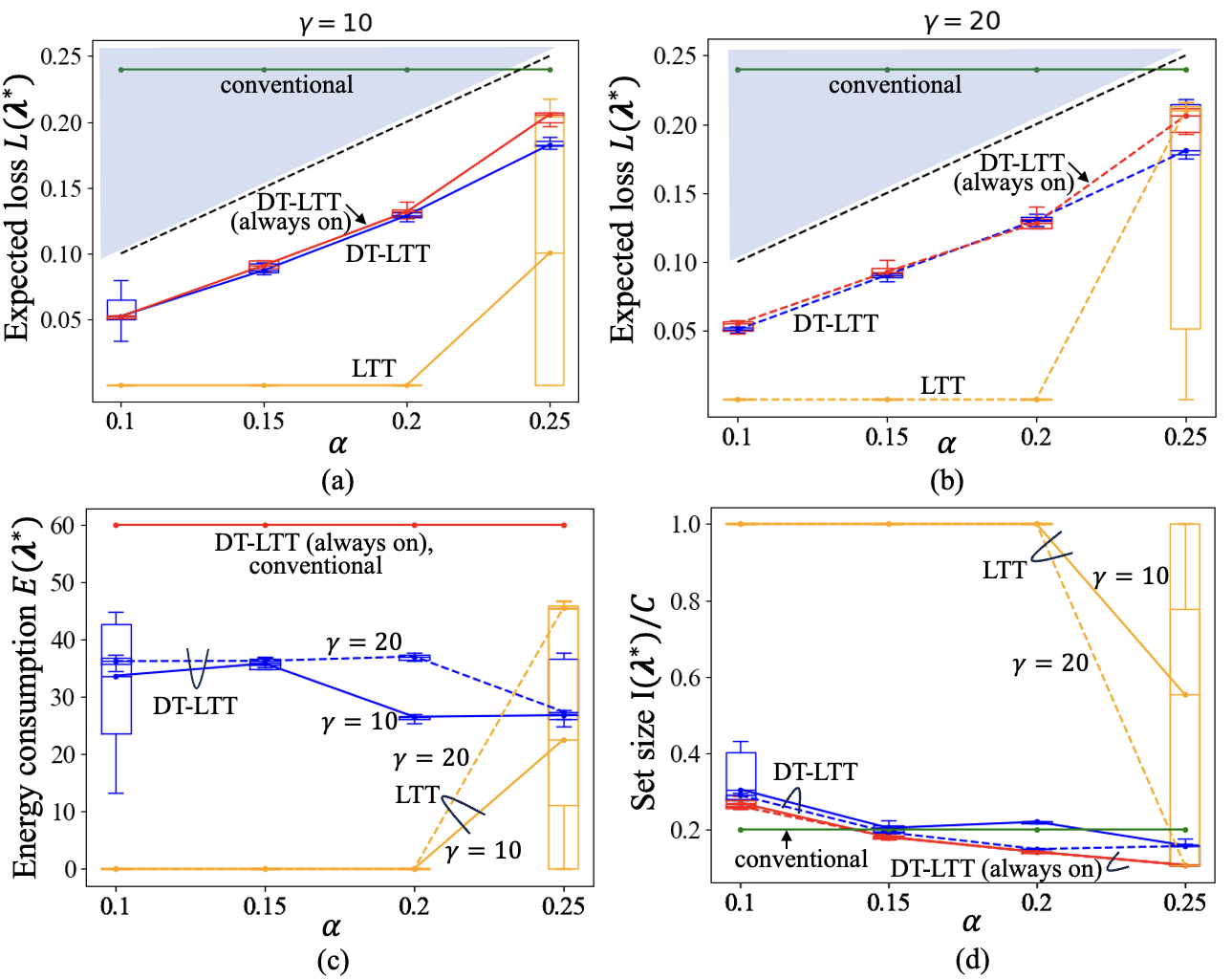}
	\caption{(a)-(b) Expected loss $L(\mv \lambda^*)$ versus the reliability target $\alpha$.  (c) Energy consumption $E(\mv \lambda^*)$ versus the reliability target $\alpha$. (d) Average normalized predicted set size $I(\mv \lambda^*)/C$ as a function of the reliability target $\alpha$ (with $L^{\rm max}=60$ and  $\delta=0.05$). }
	\label{LTT}
\end{figure*} 

In Fig.~\ref{LTT}, we validate the reliability, energy consumption and informativeness of the decisions produced by the calibrated system as a function of the target miscoverage loss $\alpha$ with $\delta=0.05$.  The ground-truth expected loss $L(\mv \lambda^*)$, energy consumption $E(\mv \lambda^*)$ and set size $I(\mv \lambda^*)$ are obtained by averaging over the test set. In Fig.~\ref{LTT}(a)-(b), the shaded area corresponds to average miscoverage losses that do not satisfy the average constraint \eqref{eq:goal}. In a manner consistent with Theorem 1, we fix a single realization of dataset $\mathcal{D}^{\rm DT}$, simulated channels at the digital twin, and real channels, and evaluate the variability of expected loss, energy consumption, and normalized set size with respect to the realization of dataset $\mathcal{D}^{\rm PT}$. Specifically, each box spans the interquartile range of the corresponding random quantity, with a line indicating the median, while the whiskers extend from the box to show the overall range of the observed values. 

From Fig.~\ref{LTT}(a) and Fig.~\ref{LTT}(b), the conventional calibration scheme fails to meet the reliability requirement, while the basic LTT scheme selects conservative hyperparameters for $\alpha=0.1$, $\alpha=0.15$ and $\alpha=0.2$, by including all classes in the predicted set, leading to zero expected loss. In contrast, the proposed DT-LTT schemes are guaranteed to meet the probabilistic reliability requirement \eqref{eq:goal}  as per Theorem 1. Furthermore,  as the allowed miscoverage probability $\alpha$ increases, the expected loss obtained with DT-LTT also grows accordingly. 

Looking now at the bottom part of Fig.~\ref{LTT}, it is observed that  the DT-LTT  scheme with an always-on receiver is over-conservative, yielding a large energy consumption, which does not adapt to varying reliability requirements $\alpha$ (Fig.~\ref{LTT}(c)). This is because this scheme is not given the freedom to keep the main radio of the receiver off in an adaptive manner. In contrast, DT-LTT is able to adjust the energy consumption to the tolerated unreliability level  $\alpha$, reducing the energy consumption accordingly.

The reduction in energy consumption afforded by a larger value of $\alpha$ depends on the design parameter $\gamma$, which dictates the relative importance of decreasing the predicted set size. In particular, increasing $\gamma$ cause the DT-LTT calibration schemes to  further reduce the set size  as $\alpha$ increases, as a smaller set can support a larger miscoverage rate $\alpha$. In this regard, for DT-LTT with $\gamma=10$, the set size initially decreases and then increases with $\alpha$. This is due to the importance attributed by calibration to lowering energy consumption, which calls for a larger predicted set to meet the reliability condition. Conversely, with $\gamma=20$, the set size consistently decreases with $\alpha$, as the primary objective is to minimize the set size. 

We have also carried out experiments with the DVS128 Gesture dataset and the performance results are qualitatively very similar to Fig. 6, and thus we have decided not to include them due to lack of space.

\subsection{Impact of Digital Twin Fidelity}
In practice, the digital twin may employ simplified or approximated models of the physical system due to computational limitations or modeling errors. In this subsection, we evaluate the impact of a mismatch between the ground-truth physical system and the digital twin model. To this end, in this experiment, the true channel is generated by using  ray tracing in a street canyon scene with cars by following Nvidia's Sionna \cite{sionna}. In contrast, the digital twin model assumes the standard tapped delay line channel model from the 3GPP TR38901 specification with a variable number of paths $N^{\rm P}_{\rm DT}$  \cite{sionna}.  Consequently, the  digital twin uses a mismatched simulator,  which follows a statistical model, rather than one that is adapted to the geometry under which the real channels are generated via ray tracing. The level of real-to-simulation mismatch can be partly controlled via the choice of the  number of paths $N^{\rm P}_{\rm DT}$. Furthermore, we also show the performance of DT-LTT when using a channel model matched to the real channels.  We set $\alpha=0.2$, $\delta=0.05$, and $\gamma=10$.

\begin{figure*}[t!]
	\centering
	\includegraphics[width=6.7in]{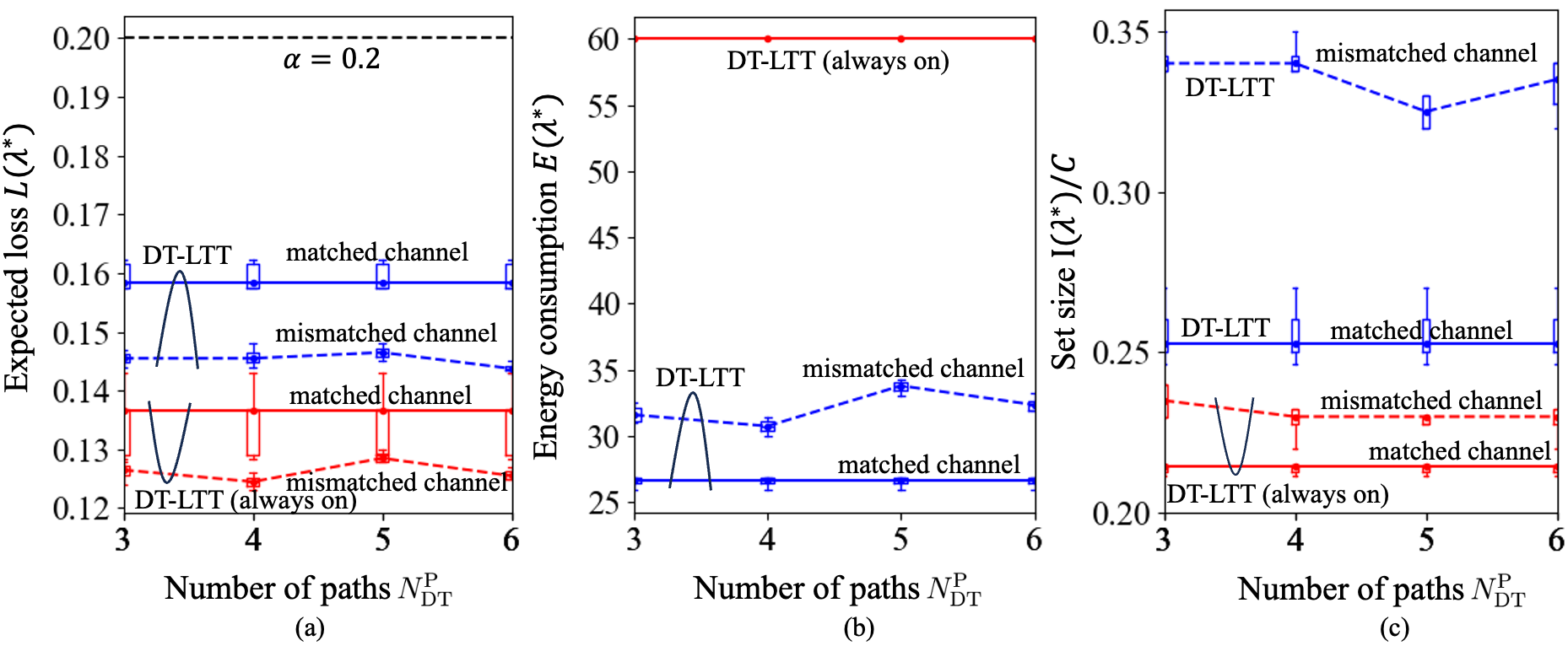}
	\caption{(a) Expected loss $L(\mv \lambda^*)$ versus the number of paths $N^{\rm P}_{\rm DT}$ of the channel simulated at digital twin. (b) Energy consumption $E(\mv \lambda^*)$ versus the number of paths $N^{\rm P}_{\rm DT}$ of the channel simulated at digital twin. (c) Average normalized predicted set size as a function of the number of paths $N^{\rm P}_{\rm DT}$ of the channel simulated at digital twin with $\alpha=0.2$, $\delta=0.05$ and $L^{\rm max}=60$.}
	\label{r4}
\end{figure*} 

In Fig.~\ref{r4}, we present the expected loss, energy consumption, and the normalized set size as a function of the number of paths $N^{\rm P}_{\rm DT}$ in the simulated channel in digital twin. As shown in Fig.~\ref{r4}(a), DT-LTT ensures the reliability condition \eqref{eq:goal} irrespective of the fidelity of the digital twin. Furthermore, as seen in Fig.~\ref{r4}(b), higher energy is required for mismatched DT model in order to achieve the reliability condition. Finally, as illustrated in Fig.~\ref{r4}(c), a richer DT model, with a larger number of paths, supports the selection of hyperparameters that reduce the set size, improving the informativeness of the decision at the receiver.

\section{Conclusions} \label{con}
This paper has introduced a novel architecture that integrates wake-up radios into a split neuromorphic computing system. A key challenge in this integration lies in determining thresholds for sensing, WUS detection, and decision-making processes so that the system maintains an expected decision-making loss below a pre-defined target level. To tackle this problem, we have proposed a digital twin-based calibration algorithm that ensures the reliability of the receiver's decision, while also optimizing a desired trade-off between energy consumption and informativeness of the decision. By leveraging a digital twin of the system, the use of on-air resources for calibration is reduced. Experimental results demonstrated the effectiveness of the proposed algorithm, confirming the theoretical guarantees on reliability, which hold irrespective of the data distribution and of the fidelity of the digital twin.

Future research may explore a hardware-based evaluation of the proposed solution, encompassing integrated sensing, computation, and communication \cite{ke2024neuromorphic}. In terms of algorithm extensions, future work may consider incorporating delay-adaptive decision making by producing an early output once the system is confident in the inference results \cite{chen2023spikecp, chen2024agreeing}.

\section*{Appendix: Proof of Theorem 1} 
The reliability condition \eqref{theo} is a consequence of the properties of LTT \cite{angelopoulos2021learn}, which is leveraged by DT-LTT via the Pareto testing method introduced in \cite{laufer2022efficiently}. As detailed next, LTT formulates the problem of hyperparameters selection in the framework of multiple-hypothesis testing. 

Consider first a single hyperparameter vector $\mv \lambda$, and define  the null hypothesis
\begin{align}
    \mathcal{H}(\mv \lambda): L^{\rm PT}(\mv \lambda) > \alpha 
\end{align} 
that the hyperparameter vector $\mv \lambda$ does not guarantee the desired reliability level $\alpha$, where $L^{\rm PT}(\mv \lambda)\in[0,1]$ is assumed to be bounded. Rejecting hypothesis $\mathcal{H}(\mv \lambda)$ implies that the calibration algorithms deems that the hyperparameter vector $\mv \lambda$ ensures the reliability condition $L^{\rm PT}(\mv \lambda) \leq \alpha$ in \eqref{eq:goal}.

To decide whether to accept or reject the null hypothesis $\mathcal{H}(\mv \lambda)$, one can evaluate  a p-value associated with hypothesis $\mathcal{H}(\mv \lambda)$, such as 
\begin{align}
    p(\mv \lambda) =e^{-2|\mathcal{D}^{\rm PT}|(\alpha-\hat{L}^{\rm PT}(\scalebox{0.8}{\mv \lambda}))^2_{+}}. \label{pvalue}
\end{align}
The quantity \eqref{pvalue} is indeed a valid p-value for the null hypothesis $\mathcal{H}(\mv \lambda)$ since the probability 
\begin{align}
    \Pr[p(\mv \lambda)\leq \delta] \leq \delta  \label{pvalid}
\end{align}
holds for $\delta\in[0,1]$, with the probability $\Pr[\cdot]$ evaluated with respect to the distribution of dataset $\mathcal{D}^{\rm PT}$ and the true channels  $\{\mathbf{h}_n \sim p(\mathbf{h})\}_{n=1}^{|\mathcal{D}^{\rm PT}|}$ under the null hypothesis $\mathcal{H}(\mv \lambda)$. The inequality \eqref{pvalid} is  verified by Hoeffding’s inequality due to the boundedness of the assumed loss \cite{angelopoulos2021learn}. 

Plugging \eqref{pvalue} into \eqref{pvalid}, the inequality \eqref{pvalid} is equivalent to the condition $\Pr[\hat{L}^{\rm PT}(\mv \lambda) \leq \psi(\alpha, \delta)]\leq \delta$ for any fixed hyperparameter $\mv \lambda$. Therefore, if the inequality $\hat{L}^{\rm PT}(\mv \lambda) \leq \psi(\alpha, \delta)$ is verified, so is the required reliability condition \eqref{eq:goal}.

The discussion so far has focused on a single hyperparameter $\mv \lambda$. However, DT-LTT considers multiple hypotheses $\mathcal{H}(\mv \lambda)$ corresponding to different candidate hyperparameter vectors  $\mv \lambda$. To this end, DT-LTT follows fixed sequence testing via Pareto testing \cite{laufer2022efficiently}. Accordingly, the hyperparameter vectors are tested sequentially stopping as soon as the first hyperparameter vector $\mv \lambda$ is found for which hypothesis $\mathcal{H}(\mv \lambda)$ is accepted. By \cite[Algorithm 1]{angelopoulos2021learn}, this guarantees that all the hyperparameters associated with the rejected hypotheses ensure the reliability condition $L^{\rm PT}(\mv \lambda) \leq \alpha$  with probability at least $1-\delta$. Finally, the conservative hyperparameter $\mv \lambda=[\lambda^{\rm s}=\infty, \lambda^{\rm w}=\infty, \lambda^{\rm d}=\infty]$ also satisfies the reliability condition \eqref{eq:goal}, since the predicted set $\mathcal{C}$ always includes the true label, concluding the proof.

\small{
\bibliographystyle{IEEEtran}
\bibliography{references}}

\end{document}